\documentclass[showpacs,eqsecnum,superscriptaddress,nofootinbib,preprint,11pt]{revtex4}
\usepackage{bm} 
\usepackage{amsmath} 
\usepackage{amssymb} 
\usepackage{amsthm} 

\def\boldpi{\boldsymbol{\pi}}
\def\boldsigma{\boldsymbol{\sigma}}

\def\sdotp{\boldsigma\cdot\boldpi}
\def\sdotE{\boldsigma\cdot\mathbf{E}}
\def\sdotB{\boldsigma\cdot\mathbf{B}}
\def\Edotp{\mathbf{E}\cdot\boldpi}
\def\Bdotp{\mathbf{B}\cdot\boldpi}
\def\Etimesp{\mathbf{E}\times\boldpi}

\def\calH{\mathcal{H}}
\def\calX{\mathcal{X}}
\def\calY{\mathcal{Y}}
\def\calZ{\mathcal{Z}}

\def\mbeta{\tilde{\beta}}
\def\malpha{\tilde{\alpha}}
\def\mbalpha{\tilde{\boldsymbol{\alpha}}}

\def\mbsigma{\tilde{\boldsymbol{\sigma}}}
\def\mgamma{\tilde{\gamma}}

\def\bigO{O}

\newcommand{\secref}[1]{Sec.~\ref{#1}}
\newcommand{\eqnref}[1]{(\ref{#1})}

\newcommand{\thmref}[1]{Theorem~\ref{#1}}

\theoremstyle{definition}
\newtheorem{theorem}{Theorem}

\begin{document}

\title
{Exact Foldy-Wouthuysen transformation of the Dirac-Pauli Hamiltonian in the weak-field limit by the method of direct perturbation theory}

\author{Dah-Wei Chiou}
\email{dwchiou@gmail.com}
\affiliation{Department of Physics, National Taiwan Normal University, Taipei 11677, Taiwan}
\affiliation{Department of Physics and Center for Condensed Matter Sciences, National Taiwan University, Taipei 10617, Taiwan}

\author{Tsung-Wei Chen}
\email{twchen@mail.nsysu.edu.tw} \affiliation{Department of Physics, National Sun Yat-sen University, Kaohsiung 80424, Taiwan}


\begin{abstract}
We apply the method of direct perturbation theory for the Foldy-Wouthuysen (FW) transformation upon the Dirac-Pauli Hamiltonian subject to external electromagnetic fields. The exact FW transformations exist and agree with those obtained by Eriksen's method for two special cases. In the weak-field limit of static and homogeneous electromagnetic fields, by mathematical induction on the orders of $1/c$ in the power series, we rigorously prove the long-held speculation: the FW transformed Dirac-Pauli Hamiltonian is in full agreement with the classical counterpart, which is the sum of the orbital Hamiltonian for the Lorentz force equation and the spin Hamiltonian for the Thomas-Bargmann-Michel-Telegdi equation.
\end{abstract}


\maketitle

\tableofcontents
\newpage

\section{Introduction}
The relativistic quantum theory for a spin-$1/2$ particle is described by the Dirac equation  \cite{Dirac1928,Dirac1982}, which, in the rigorous sense, is self-consistent only in the context of quantum field theory as particle-antiparticle pairs can be created and annihilated. The question that naturally arises is whether in the low-energy limit the particle and antiparticle can be treated separately without taking into account the field-theory interaction between them on the grounds that the probability of particle-antiparticle pair creation and annihilation is negligible. It turns out that such separation is possible and indeed gives an adequate description of the relativistic quantum dynamics whenever the relevant energy (the particle's energy interacting with external, e.g., electromagnetic, fields) is much smaller than the Dirac energy gap $2mc^2$ ($m$ is the particle's mass).

The Foldy-Wouthuysen (FW) transformation is the method devised to achieve the particle-antiparticle separation via a series of successive unitary transformations, each of which block-diagonalizes the Dirac Hamiltonian to a certain order of $1/m$ \cite{Foldy1950} (see \cite{Strange2008} for a review). In the same spirit of the standard FW method, many different approaches have been developed for various advantages \cite{Lowding1951, Luttinger1955, Eriksen1958, Douglas1974, Hess1985, Hess1986, Rutkowski1986a, Rutkowski1986b, Rutkowski1986c, Heully1986, Kutzelnigg1990, Silenko2013-analysis, Winkler2003, Silenko2003, Reiher2004a, Reiher2004b, Bliokh2005, Goss2007, Goss2007b, Silenko2008, Peng2012} (also see \cite{Reiher2009 Book} for a review in the context of relativistic quantum chemistry). Particularly, the works by Rutkowski \cite{Rutkowski1986a,Rutkowski1986b,Rutkowski1986c} and Heully \cite{Heully1986} proposed and exploited a self-consistent equation that allows one to obtain the block-diagonalized Dirac Hamiltonian without explicitly evoking decomposition of even and odd Dirac matrices; the perturbation approach developed by Rutkowski is now known as direct perturbation theory (DPT).

Furthermore, to phenomenologically account for any presence of the anomalous magnetic moment, the Dirac equation, augmented with extra terms explicitly dependent on electromagnetic field strength, is extended to the Dirac-Pauli equation to describe the relativistic quantum dynamics of a spin-$1/2$ particle of which the gyromagnetic ratio is different from $q/(mc)$ ($q$ is the particle's charge) \cite{Pauli1941}. The FW methods for the Dirac equation can be straightforwardly carried over to the Dirac-Pauli equation without much difficulty \cite{Silenko2008}.

On the other hand, the classical (non-quantum) dynamics for a relativistic point particle endowed with charge and intrinsic spin in electromagnetic fields is well understood. The orbital motion is governed by the Lorentz force equation and the precession of spin by the Thomas-Bargmann-Michel-Telegdi (T-BMT) equation \cite{Thomas1927,Bargmann1959} (see Chapter 11 of \cite{Jackson1999} for a review). The orbital Hamiltonian for the Lorentz force equation plus the spin Hamiltonian for the T-BMT equation provides a low-energy description of the relativistic spinor dynamics. It is natural to conjecture that, in the weak-field limit of external electromagnetic fields, the Dirac or, more generically, the Dirac-Pauli Hamiltonian, after block diagonalization, should correspond to the sum of the classical orbital and spin Hamiltonians.

This quantum-classical correspondence between the Dirac equation and the Lorentz force equation along with the T-BMT equation is crucial to the problem of finding and interpreting \emph{spin operators} for the Dirac equation --- a problem which has been discussed in the literature for a long time but remains challenging and unsolved in the presence of external fields (see Section 2.4 of \cite{Bagrov Book} and references therein for more discussions). Validity of the correspondence has been investigated from different aspects with various degrees of rigor \cite{Silenko2008,Rubinow1963,Rafa1964,Froh1993,Silenko1995} and explicated in \cite{Chen2014}. In the case of static and homogeneous electromagnetic fields, it has been shown that the FW transformed Dirac-Pauli Hamiltonian is in agreement with the classical Hamiltonian up to the order of $1/m^8$, if nonlinear terms of electromagnetic fields are neglected in the weak-field limit \cite{Chen2010}. Recently, the work of \cite{Chen2010} was extended to the order of $1/m^{14}$ by applying the method of DPT, cast in the style of Kutzelnigg's implementation \cite{Kutzelnigg1990} with a further simplification scheme introduced \cite{Chen2013}.

Although the result of \cite{Chen2013} is very impressive, the long sought-after proof for the full agreement to any arbitrary order is still missing. Thanks to the result obtained in \cite{Chen2013} up to the high order of $1/m^{14}$, we are now able to conjecture the generic expression for terms of any given order in the DPT method and then give a proof by mathematical induction on the orders of power series expansion.\footnote{Various prior works in different approaches have provided algorithms of automated generation of arbitrarily high order terms in the order-by-order expansion (e.g.\ see \cite{Reiher2004a,Reiher2004b}). The method adopted in \cite{Chen2013} can be programmed as an automated algorithm as well, but automation is not very necessary for our purpose because in the end the proof of mathematical induction will ascertain the analytical form of terms in \emph{any} orders.} In this paper, we elaborate on Kutzelnigg's implementation of DPT and present the rigorous proof of the quantum-classical correspondence. As a secondary result, we also show that the exact FW transformations by the DPT method exist and agree with those obtained by Eriksen's method \cite{Eriksen1958} for two special cases of arbitrary magnetostatic field and arbitrary electrostatic field. Various conceptual issues of the FW transformation are also addressed and clarified.\footnote{It should be emphasized that the main purpose of this paper is to provide a rigorous proof of the quantum-classical correspondence. Although some other conceptual issues are also addressed, it is \emph{not} our intent to take part in the debate on mathematical rigor and legitimacy of the FW transformation (see \secref{sec:remarks} for more comments).}

This paper is organized as follows. After briefly reviewing the classical and Dirac-Pauli spinors in \secref{sec:classical spinor} and \secref{sec:Dirac-Pauli spinor}, respectively, we look into the FW transformation with the emphasis on Kutzelnigg's method of DPT in \secref{sec:FW transform}.\footnote{These parts deliberately contain some of the same review materials in \cite{Chen2014}.} We then present the proof for the exact quantum-classical correspondence in the weak-field limit for the Dirac Hamiltonian in \secref{sec:Dirac Hamiltonian} and then for the Dirac-Pauli Hamiltonian in \secref{sec:Dirac-Pauli Hamiltonian}.\footnote{The proof is schematically summarized in a separate article \cite{Chiou:2015pca}, which is much shorter and may be more readable for those who do not intend to know the details.} Conclusions are summarized and discussed in \secref{sec:summary}.

\section{Classical relativistic spinor}\label{sec:classical spinor}
In this section, we briefly review the classical dynamics of a classical relativistic spinor, which is detailed in \cite{Chen2014}.

For a relativistic point particle endowed with electric charge $q$ and intrinsic spin $\mathbf{s}$ subject to external electromagnetic fields $\mathbf{E}$ and $\mathbf{B}$ (the corresponding 4-potential is denoted as $A^\mu=(\phi,\mathbf{A})$ and the electromagnetic tensor by $F_{\mu\nu}$), the orbital motion, which is governed by the Lorentz force equation, and the spin precession, which is governed by the T-BMT equation, are simultaneously described by the total Hamiltonian
\begin{equation}\label{H cl}
H(\mathbf{x},\mathbf{p},\mathbf{s};t)=H_\mathrm{orbit}(\mathbf{x},\mathbf{p};t)
+H_\mathrm{spin}(\mathbf{s},\mathbf{x},\mathbf{p};t)+\bigO(F_{\mu\nu}^2,\hbar^2)
\end{equation}
with the orbital Hamiltonian given by
\begin{equation}\label{H orbit}
H_\mathrm{orbit}(\mathbf{x},\mathbf{p};t)=\sqrt{m^2c^4+c^2\boldpi^2}+q\phi(\mathbf{x},t)
\end{equation}
and the spin Hamiltonian given by
\begin{eqnarray}\label{H spin}
H_\mathrm{spin}(\mathbf{s},\mathbf{x},\mathbf{p};t)
&=&-\mathbf{s}\cdot\left[
\left(\gamma'_m+\frac{q}{mc}\frac{1}{\gamma_{\boldpi}}\right)\mathbf{B}(\mathbf{x})
-\gamma'_m\frac{1}{\gamma_{\boldpi}(1+\gamma_{\boldpi})}
\left(\frac{\boldpi}{mc}\cdot\mathbf{B}(\mathbf{x})\right)\frac{\boldpi}{mc}
\right.\nonumber\\
&&\qquad\qquad
\left.
\mbox{}-\left(\gamma'_m\frac{1}{\gamma_{\boldpi}} +\frac{q}{mc}\frac{1}{\gamma_{\boldpi}(1+\gamma_{\boldpi})}\right)
\left(\frac{\boldpi}{mc}\times\mathbf{E}(\mathbf{x})\right)
\right],
\end{eqnarray}
where the kinematic momentum $\boldpi$ is defined as
\begin{equation}
\boldpi:=\mathbf{p}-\frac{q}{c}\,\mathbf{A}(\mathbf{x},t),
\end{equation}
the Lorentz factor associated with $\boldsymbol{\pi}$ is defined as
\begin{equation}\label{gamma pi cl}
\gamma_{\boldsymbol{\pi}}:=\sqrt{1+\left(\frac{\boldsymbol{\pi}}{mc}\right)^2}\,,
\end{equation}
and $\gamma'_m$ is the \emph{anomalous gyromagnetic ratio}
\begin{equation}
\gamma'_m:=\gamma_m-\frac{q}{mc}
\end{equation}
with $\gamma_m$ being the total gyromagnetic ratio.

It should be remarked that the classical theory described by \eqnref{H cl} respects Lorentz invariance only within a high degree of accuracy, unless the terms of $\bigO(F_{\mu\nu}^2,\hbar^2)$ are appropriately supplemented by a more fundamental quantum theory such as the Dirac-Pauli theory. In the weak-field limit, the nonlinear electromagnetic corrections of $\bigO(F_{\mu\nu}^2)$ can be neglected, and the particle's velocity is given by
\begin{eqnarray}
\mathbf{v}\equiv\frac{d\mathbf{x}}{dt}
=\boldsymbol{\nabla}_\mathbf{p}H_\mathrm{orbit}+\boldsymbol{\nabla}_\mathbf{p}H_\mathrm{spin}
\approx\frac{\boldsymbol{\pi}}{m\gamma_{\boldsymbol{\pi}}}
\end{eqnarray}
provided
\begin{equation}\label{small H spin}
H_\mathrm{spin}\ll mc^2,
\end{equation}
which is true in the weak-field limit. Consequently, $\boldpi$ remains to be the kinematic momentum associated with $\mathbf{v}$, i.e.,
\begin{equation}\label{pi appro}
\boldpi\approx m\mathbf{U}\equiv \gamma m \mathbf{v},
\end{equation}
and $\gamma_{\boldpi}$ is to be identified with the ordinary Lorentz boost factor, i.e.,\
\begin{equation}\label{gamma pi appro}
\gamma_{\boldpi}\approx \gamma:=\frac{1}{\sqrt{1-\mathbf{v}^2/c^2}}.
\end{equation}
Furthermore, the Dirac-Pauli theory also gives rise to the Darwin term of $\bigO(\hbar^2)$, which has no classical (non-quantum) correspondence and does not show up in the case of homogeneous fields.

\section{Dirac-Pauli spinor}\label{sec:Dirac-Pauli spinor}
The relativistic quantum theory of a spin-$1/2$ particle subject to external electromagnetic fields is described by the Dirac equation \cite{Dirac1928,Dirac1982}
\begin{equation}\label{Dirac eq}
\mgamma^\mu D_\mu|\psi\rangle+i\frac{mc}{\hbar}|\psi\rangle=0,
\end{equation}
where the Dirac bispinor $|\psi\rangle=(\chi,\varphi)^T$ is composed of two 2-component Weyl spinors $\chi$ and $\varphi$, the covariant derivative $D_\mu$ is given by
\begin{eqnarray}
D_\mu&:=&\partial_\mu +\frac{iq}{\hbar c}A_\mu \equiv -\frac{i}{\hbar}\pi_\mu :=-\frac{i}{\hbar}\left(p_\mu-\frac{q}{c}A_\mu\right)\nonumber\\
&=&\left(\frac{1}{c}\frac{\partial}{\partial t}+\frac{iq}{\hbar c}\phi,\,\boldsymbol{\nabla}-\frac{iq}{\hbar c}\mathbf{A}\right)
\equiv -\frac{i}{\hbar}\left(\frac{E-q\,\phi}{c}, \,-\left(\mathbf{p}-\frac{q}{c}\mathbf{A}\right)\right)
\end{eqnarray}
with $p^\mu=(E/c,\mathbf{p})$ being the 4-vector of canonical energy and momentum and $\pi^\mu=(W/c,\boldsymbol{\pi})$ being the 4-vector of kinematic energy and momentum, and $\mgamma^\mu$ are $4\times4$ matrices\footnote{Throughout this paper, a tilde is attached to denote a $4\times4$ matrix.} that satisfy
\begin{equation}
\mgamma^\mu\mgamma^\nu+\mgamma^\nu\mgamma^\mu = 2g^{\mu\nu}.
\end{equation}
The Dirac equation gives rise to the magnetic moment with $\gamma_m=q/(mc)$ (i.e., the $g$-factor is given by $g=2$). To incorporate any anomalous magnetic moment $\mu'=\gamma'_m\hbar/2$, one can modify the Dirac equation to the Dirac-Pauli equation with augmentation of explicit dependence on field strength \cite{Silenko2008,Pauli1941}:
\begin{equation}\label{Dirac-Pauli eq}
\mgamma^\mu D_\mu|\psi\rangle+i\frac{mc}{\hbar}|\psi\rangle
+\frac{i\mu'}{2c}\mgamma^\mu\mgamma^\nu F_{\mu\nu}|\psi\rangle=0.
\end{equation}

The Pauli-Dirac equation can be cast in the Hamiltonian formalism as
\begin{equation}\label{Dirac-Pauli eq w H}
i\hbar\frac{\partial}{\partial t}|\psi\rangle = \tilde{\calH}|\psi\rangle
\end{equation}
with the Dirac Hamiltonian $\tilde{H}$ and the Dirac-Pauli Hamiltonian $\tilde{\calH}$ defined as
\begin{subequations}\label{H Dirac and Dirac-Pauli}
\begin{eqnarray}
\label{H Dirac}
\tilde{H} &=& mc^2\mbeta+c\,\mbalpha\cdot\left(\mathbf{p}-\frac{q}{c}\mathbf{A}\right)+q\phi
\equiv
\left(
  \begin{array}{cc}
    mc^2+q\phi & c\,\boldsigma\cdot\boldpi \\
    c\,\boldsigma\cdot\boldpi & -mc^2+q\phi \\
  \end{array}
\right)
,\\
\label{H Dirac-Pauli}
\tilde{\calH} &=& \tilde{H} +\mu'\left(-\mbeta\mbsigma\cdot\mathbf{B}+i\mbeta\mbalpha\cdot\mathbf{E}\right)
\equiv
\left(
  \begin{array}{cc}
    mc^2+q\phi-\mu'\sdotB & c\,\boldsigma\cdot\boldpi+i\mu'\sdotE \\
    c\,\boldsigma\cdot\boldpi-i\mu'\sdotE & -mc^2+q\phi+\mu'\sdotB \\
  \end{array}
\right),
\qquad
\end{eqnarray}
\end{subequations}
where the $4\times4$ matrices are given explicitly by
\begin{equation}
\mbeta=\left(
         \begin{array}{cc}
           \openone & 0 \\
           0 & -\openone \\
         \end{array}
       \right),\
\mbalpha=\left(
           \begin{array}{cc}
             0 & \boldsigma \\
             \boldsigma & 0 \\
           \end{array}
         \right),\
\mbsigma=\left(
           \begin{array}{cc}
             \boldsigma & 0 \\
             0 & \boldsigma \\
           \end{array}
         \right),
\end{equation}
and $\boldsigma=(\sigma_x,\sigma_y,\sigma_z)$ are the $2\times2$ Pauli matrices.
Accordingly, the $\mgamma$ matrices are given by
\begin{equation}
\mgamma^0=\mbeta,
\quad
\mgamma^i=\mbeta\malpha^i=\left(
           \begin{array}{cc}
             0 & \sigma_i \\
             -\sigma_i & 0 \\
           \end{array}
         \right).
\end{equation}

\section{Foldy-Wouthuysen transformation}\label{sec:FW transform}
The Dirac or Dirac-Pauli Hamiltonians \eqnref{H Dirac and Dirac-Pauli} (or, more generally, with other corrections) can be schematically put in the form
\begin{equation}
\tilde{\calH}=\mbeta mc^2+\tilde{\mathcal{O}}+\tilde{\mathcal{E}},
\end{equation}
where $\tilde{\mathcal{E}}$ is the ``even'' part that commutes with $\mbeta$, i.e., $\mbeta\tilde{\mathcal{E}}\mbeta=\tilde{\mathcal{E}}$, while $\tilde{\mathcal{O}}$ is the ``odd'' part that anticommutes with $\mbeta$, i.e., $\mbeta\tilde{\mathcal{O}}\mbeta=-\tilde{\mathcal{O}}$. Because of the presence of the odd part, the Hamiltonian in the Dirac bispinor representation is not block-diagonalized, and thus the particle and antiparticle components are entangled in each of the Weyl spinors $\chi$ and $\varphi$. The question that naturally arises is whether we can find a representation in which the particle and antiparticle are separated, or equivalently, the Hamiltonian is block-diagonalized. Foldy and Wouthuysen have shown that such a representation is possible \cite{Foldy1950,Strange2008}. The Foldy-Wouthuysen (FW) transformation is a unitary and nonexplicitly time-dependent transformation on the Dirac bispinor
\begin{subequations}
\begin{eqnarray}
|\psi\rangle&\rightarrow&|\psi_\mathrm{FW}\rangle=\tilde{U}|\psi\rangle,\\
\label{H'}
\tilde{\calH}&\rightarrow&
\tilde{\calH}_\mathrm{FW}=\tilde{U}\tilde{\calH}\tilde{U}^\dag,
\end{eqnarray}
\end{subequations}
which leaves \eqnref{Dirac-Pauli eq w H} in the form
\begin{equation}\label{H' psi'}
i\hbar\frac{\partial}{\partial t}|\psi_\mathrm{FW}\rangle = \tilde{\calH}_\mathrm{FW}|\psi_\mathrm{FW}\rangle
\end{equation}
and block-diagonalizes the Hamiltonian, i.e.,
$[\mbeta,\tilde{\calH}_\mathrm{FW}]=0$.
As the FW transformation separates the particle and antiparticle components, the two diagonal blocks of $\tilde{\calH}_\mathrm{FW}$ are adequate to describe the relativistic quantum dynamics of the spin-1/2 particle and antiparticle respectively.\footnote{If $\tilde{U}$ is explicitly time-dependent, instead of \eqnref{H'}, the diagonalized Hamiltonian is given by
$\tilde{\calH}_\mathrm{FW}=\tilde{U}\tilde{\calH}\tilde{U}^\dag -i\hbar\, \tilde{U}\frac{\partial}{\partial t}\tilde{U}^\dag$, which is beyond the scope of the standard FW scenario. Throughout this paper, we consider only the case in \emph{static} fields. For the nonstandard FW transformation involving non-static fields, see \cite{Silenko2014} for more details.}

However, it should be remarked that, rigorously, the Dirac equation is self-consistent only in the context of quantum field theory, in which the particle-antiparticle pairs can be created and annihilated. On this account, it might not be legitimate to block-diagonalize the Dirac Hamiltonian or its phenomenological extension such as the Dirac-Pauli Hamiltonian. In fact, some doubts have been thrown on the mathematical rigor of the FW transformation \cite{Thaller1992,Brouderr1996} (but also see \cite{Silenko2015} for discussion on its validity). If the unitary FW transformation does not exist after all, the power series used in any order-by-order methods does not converge and high-order terms might be misleading and disagree with those obtained by different methods.\footnote{For example, for the Dirac theory in the presence of both electric and magnetic fields, the term of order $F_{\mu\nu}^2$ in the method of DPT is given by $-\frac{q^2\hbar^2}{8m^3c^4}\mathbf{B}^2$, while it is given by $\frac{q^2\hbar^2}{8m^3c^4}\left(\mathbf{E}^2-\mathbf{B}^2\right)$ in the standard FW method (see \cite{Chen2013}). (Nevertheless, these two methods agree with each other on the terms linear in $F_{\mu\nu}$).} However, as will be shown in Sections \ref{sec:special case I} and \ref{sec:special case II}, the exact FW transformation does exist at least for two special cases, suggesting that particle-antiparticle separation is consistent and does not lead to any disagreement in these special situations.\footnote{As we will see shortly, the method of DPT yields exactly the same results of Eriksen's method for these two cases.} For more special cases, see \cite{Nikitin:1998}, which gives a wide class of external electromagnetic fields that admit the exact FW transformation.

Furthermore, in the regime of weak fields such that the energy interacting with electromagnetic fields does not exceed the Dirac energy gap $2mc^2$, we expect that the probability of pair creation and annihilation is negligible, and accordingly the FW transformation remains sensible and the block-diagonalized Hamiltonian is adequate to describe the relativistic quantum dynamics of the spin-$1/2$ particle and antiparticle separately without taking into account the field-theory interaction with each other. Starting from \secref{sec:weak-field limit}, this paper is mainly devoted to this topic.

It should be noted that even if the unitary FW transformation exists, it is far from unique, as one can easily perform further unitary transformations that preserve the block decomposition upon the block-diagonalized Hamiltonian. The non-uniqueness does not lead to any ambiguity, as different block-diagonalization transformations are unitarily equivalent to one another and thus yield the same physics. While the physics is the same, however, the pertinent operators $\boldsigma$, $\mathbf{x}$, and $\mathbf{p}$ may represent very different physical quantities in different representations. To figure out the operators' physical interpretations, it is crucial to compare the resulting FW transformed Hamiltonian with the classical counterpart in a certain classical limit via the correspondence principle. The comparison will be carried out explicitly in the weak-field limit for Kutzelnigg's method of DPT; it turns out that, in Kutzelnigg's method (and in fact in most FW methods in the literature), $\boldsigma$, $\mathbf{x}$, and $\mathbf{p}$ simply represent the spin, position, and conjugate momentum of the particle (as decoupled from the antiparticle) in the resulting FW representation. In other words, the method is ``minimalist'' in the sense that it does not give rise to further transformations that obscure the operators' interpretations other than block diagonalization.

There are various methods for the FW transformation with different advantages. In this paper, we adopt Kutzelnigg's implementation \cite{Kutzelnigg1990} of DPT \cite{Rutkowski1986a,Rutkowski1986b,Rutkowski1986c,Heully1986} improved with a further simplification scheme \cite{Chen2013}.

\subsection{Method of direct perturbation theory}
In Kutzelnigg's implementation \cite{Kutzelnigg1990} of DPT \cite{Rutkowski1986a,Rutkowski1986b,Rutkowski1986c,Heully1986}, the FW unitary transformation is assumed to take the form
\begin{equation}\label{U and Udag}
\tilde{U}=
\left(
  \begin{array}{cc}
    \calY & \calY\calX^\dag \\
    -\calZ\calX & \calZ \\
  \end{array}
\right),
\qquad
\tilde{U}^\dag=
\left(
  \begin{array}{cc}
    \calY & -\calX^\dag\calZ \\
    \calX\calY & \calZ \\
  \end{array}
\right),
\end{equation}
where the $2\times2$ hermitian operators $\calY$ and $\calZ$ are defined as
\begin{equation}
\calY=\calY^\dag=\frac{1}{\sqrt{1+\calX^\dag\calX}}, \qquad
\calZ=\calZ^\dag=\frac{1}{\sqrt{1+\calX\calX^\dag}}
\end{equation}
for some operator $\calX$ to be determined. It is easy to show that
\begin{equation}
\tilde{U}\tilde{U}^\dag=\left(
          \begin{array}{cc}
            \calY\left(1+\calX^\dag\calX\right)\calY & 0 \\
            0 & \calZ\left(1+\calX\calX^\dag\right)\calZ \\
          \end{array}
        \right)
        =1.
\end{equation}

Generically, we assume the Hamiltonian operator $\tilde{\calH}$ takes the form
\begin{equation}\label{H generic form}
\tilde{\calH}=
\left(
  \begin{array}{cc}
    H_+ & H_0 \\
    H_0^\dag & H_- \\
  \end{array}
\right),
\quad
\text{with}\
H_+^\dag=H_+,\ H_-^\dag=H_-,
\end{equation}
and the FW transformed Hamiltonian is then given by
\begin{eqnarray}
&&\tilde{\calH}_\mathrm{FW}
\equiv
\left(
  \begin{array}{cc}
    \calH_\mathrm{FW} & 0 \\
    0 & \bar{\calH}_\mathrm{FW} \\
  \end{array}
\right)
=\tilde{U}\tilde{\calH}\tilde{U}^\dag\\
&=&\left(
   \begin{array}{cc}
     \calY\left(H_++H_0\calX+\calX^\dag H_0^\dag+\calX^\dag H_-\calX\right)\calY &
     \calY\left(H_0-H_+\calX^\dag+\calX^\dag H_--\calX^\dag H_0^\dag\calX^\dag\right)\calZ \\
     \calZ\left(H_0^\dag-\calX H_++H_-\calX-\calX H_0\calX\right)\calY &
     \calZ\left(H_--H_0^\dag\calX^\dag-\calX H_0+\calX H_+^\dag\calX^\dag\right)\calZ \\
   \end{array}
 \right).\nonumber
\end{eqnarray}
The requirement that the off-diagonal blocks of $\tilde{\calH}_\mathrm{FW}$ vanish demands $\calX$ to satisfy
\begin{subequations}\label{conditions for X}
\begin{eqnarray}
&&H_0^\dag-\calX H_++H_-\calX-\calX H_0\calX = 0,\\
&&H_0-H_+\calX^\dag+\calX^\dag H_--\calX^\dag H_0^\dag\calX^\dag = 0,
\end{eqnarray}
\end{subequations}
and meanwhile the diagonal blocks read as
\begin{subequations}\label{HFW and barHFW 1}
\begin{eqnarray}
\calH_\mathrm{FW}
&=& \calY\left(H_++H_0\calX+\calX^\dag H_0^\dag+\calX^\dag H_-\calX\right)\calY,\\
\bar{\calH}_\mathrm{FW}
&=& \calZ\left(H_--H_0^\dag\calX^\dag-\calX H_0+\calX H_+^\dag\calX^\dag\right)\calZ,
\end{eqnarray}
\end{subequations}
which are manifestly hermitian.
Under the condition of \eqnref{conditions for X}, \eqnref{HFW and barHFW 1} can be further simplified as
\begin{subequations}\label{HFW and barHFW 2}
\begin{eqnarray}
\label{HFW and barHFW 2 a}
\calH_\mathrm{FW}
&=& \calY\left(H_++H_0\calX
+\calX^\dag\big(\calX H_++\calX H_0\calX\big)\right)\calY
=\calY\left((1+\calX^\dag\calX)\big(H_++H_0\calX\big)\right)\calY\nonumber\\
&=& \calY^{-1}\big(H_++H_0\calX\big)\calY,\\
\label{HFW and barHFW 2 b}
\bar{\calH}_\mathrm{FW}
&=& \calZ\left(H_--H_0^\dag\calX^\dag
+\calX\big(\calX^\dag H_--\calX^\dag H_0^\dag\calX^\dag\big)\right)\calZ
=\calZ\left((1+\calX\calX^\dag)\big(H_--H_0^\dag\calX^\dag\big)\right)\nonumber\\
&=&\calZ^{-1}\big(H_--H_0^\dag\calX^\dag\big)\calZ.
\end{eqnarray}
\end{subequations}

In the Dirac or Dirac-Pauli theory, the Hamiltonian \eqnref{H generic form} is explicitly given by \eqnref{H Dirac and Dirac-Pauli}. Consider the formal replacement:
\begin{equation}
\mathbf{p},\boldpi,\boldsigma,q,\mu',i
\rightarrow
-\mathbf{p},-\boldpi,-\boldsigma,-q,-\mu',-i,
\end{equation}
which corresponds to
\begin{equation}\label{H+ to H-}
H_+ \rightarrow -H_-,
\qquad
H_0 \rightarrow H_0^\dag,
\end{equation}
and accordingly, by \eqnref{conditions for X},
\begin{equation}\label{X to Xdag}
\calX \rightarrow \calX^\dag.
\end{equation}
Comparison between \eqnref{HFW and barHFW 2 a} and \eqnref{HFW and barHFW 2 b} by reference to \eqnref{H+ to H-} and \eqnref{X to Xdag} then implies
\begin{equation}\label{HFW to barHFW}
\bar{\calH}_\mathrm{FW}(\mathbf{x},\boldpi,\boldsigma;q,\mu')
=-\calH_\mathrm{FW}(\mathbf{x},-\boldpi,-\boldsigma;-q,-\mu').
\end{equation}
That is, $\bar{\calH}_\mathrm{FW}$ takes the form of $\calH_\mathrm{FW}$ by formally replacing $\boldpi,\boldsigma,q,\mu'$ with $-\boldpi,-\boldsigma,-q,-\mu'$ (which accounts for the charge conjugation) in addition to an overall minus sign (which account for the negative frequency).\footnote{Since $\bar{\calH}_\mathrm{FW}$ can be easily obtained by \eqnref{HFW to barHFW} once $\cal{H}_\mathrm{FW}$ is found, we focus only on the part of $\calH_\mathrm{FW}$ in the rest of this paper. When $\calH_\mathrm{FW}$ and $\bar{\calH}_\mathrm{FW}$ are combined to form $\tilde{\calH}_\mathrm{FW}$, the matrix $\mbeta$ will appear accordingly in the expression of $\tilde{\calH}_\mathrm{FW}$ as can be seen in Equations (3.14), (3.23), and (3.29) in \cite{Chen2014}.} (Also see \cite{Chen2014} for comments on the $CPT$ symmetries.)

For the Dirac-Pauli theory, \eqnref{conditions for X} and \eqnref{HFW and barHFW 2} read explicitly as
\begin{eqnarray}\label{condition calX}
{2mc^2}\calX&=&-\calX c\,\sdotp\,\calX+c\,\sdotp+q[\phi,\calX]\nonumber\\
&& \mbox{}-i\mu'\sdotE-i\mu'\calX\sdotE\,\calX +\mu'\{\calX,\sdotB\}
\end{eqnarray}
and
\begin{equation}\label{calHFW Kutzelnigg}
\calH_\mathrm{FW}
= mc^2 + \sqrt{1+\calX^\dag\calX}
\left(q\phi +c\,\sdotp\calX -\mu'\sdotB +i\mu'\sdotE\,\calX\right)
\frac{1}{\sqrt{1+\calX^\dag\calX}}.
\end{equation}
Particularly, for the Dirac theory, \eqnref{condition calX} and \eqnref{calHFW Kutzelnigg} reduce to (by simply setting $\mu'=0$)
\begin{equation}\label{condition X}
2mc^2X=
-Xc\,\sdotp X+c\,\sdotp+q[\phi,X]
\end{equation}
and
\begin{equation}\label{HFW Kutzelnigg}
H_\mathrm{FW} = mc^2 + \sqrt{1+X^\dag X}\,\left(q\phi+c\,\sdotp X\right)\frac{1}{\sqrt{1+X^\dag X}},
\end{equation}
where we have used the notations $X$ and $H_\mathrm{FW}$ in place of $\calX$ and $\calH_\mathrm{FW}$ when the Dirac-Pauli theory is reduced to the Dirac theory.

As caveated previously, the Hamiltonian $\tilde{\calH}$ might not be block-diagonalizable at all and on this account there is no guarantee that the operator $\calX$ satisfying \eqnref{condition calX} or $X$ satisfying \eqnref{condition X} exists. However, as we will see, $\calX$ or $X$ does exist in two special cases as well as in the case of homogeneous fields in the weak-field limit; accordingly $\tilde{\calH}$ is block-diagonalizable in these situations.

\subsection{Special case I}\label{sec:special case I}
As the first special case, let us consider a Dirac spinor ($\mu'=0$) with charge $q$ subject to a static magnetic field ($\partial_t\mathbf{B}=0$, $\partial_t\mathbf{A}=0$) but with no electric field ($\mathbf{E}=0$, $\phi=0$). The condition \eqnref{condition X} becomes a quadratic equation in $X$:
\begin{equation}
2mc^2X=-Xc\,\sdotp X+c\,\sdotp,
\end{equation}
which admits an exact solution
\begin{equation}\label{X case I}
X=X^\dag=\frac{c\,\sdotp}{mc^2+\sqrt{m^2c^4+c^2(\sdotp)^2}}.
\end{equation}
Equation~\eqnref{HFW Kutzelnigg} with $\phi=0$ then yields
\begin{equation}\label{HFW case I}
H_\mathrm{FW}=mc^2+c\,\sdotp X
=\sqrt{m^2c^4+c^2(\sdotp)^2}
=\sqrt{m^2c^4+c^2\boldpi^2-q\hbar\,c\sdotB}
\end{equation}
by \eqnref{identity 3}. The resulting FW transformed Hamiltonian in \eqnref{HFW case I} is exactly the same as that obtained by Eriksen's method \cite{Eriksen1958,Chen2014}.

The fact that the Dirac Hamiltonian in a static magnetic field can be block-diagonalized suggests that it is legitimate to ignore creation or annihilation of particle-antiparticle pairs. In fact, it has been shown that, in the context of QED, the charged particle-antiparticle pairs are \emph{not} produced by any static magnetic field no matter how strong the field strength is, since the instanton actions for tunneling probability for pair production are infinite \cite{Kim:2000un,Kim:2011cx}.\footnote{However, when the magnetic field changes in time, particle-antiparticle pairs can be produced \cite{Kim:2013lda}, but this situation is beyond the scope of the standard FW scenario, in which $\tilde{U}$ is assumed to be nonexplicitly time-dependent.}

If we turn off both electric and magnetic fields, \eqnref{HFW case I} reduces to
\begin{equation}
H_\mathrm{FW}=\sqrt{m^2c^4+c^2\mathbf{p}^2}\,,
\end{equation}
which is the FW transformed Hamiltonian of a free particle.

Another interesting case is of a massless spinor. When it is subject only to a static magnetic field or it carries no charge ($q=0$, such as a massless neutrino), \eqnref{X case I} with $m=0$ yields $X=X^\dag=1$, which follows from \eqnref{U and Udag} that
\begin{equation}\label{U case massless}
\tilde{U}=\frac{1}{\sqrt{2}}
\left(
  \begin{array}{cc}
    1 & 1 \\
    -1 & 1 \\
  \end{array}
\right).
\end{equation}
The trivial FW transformation \eqnref{U case massless} is nothing but the unitary transformation that transforms the Dirac basis to the Weyl basis.\footnote{In the Weyl basis, it is well known that the upper two components are decoupled from the lower two components for an uncharged massless spinor.}

Also see \cite{Nikitin:1998} and references therein for more discussions on the exact FW transformation.

\subsection{Special case II}\label{sec:special case II}
As the second special case, let us consider a Dirac-Pauli spinor with zero charge ($q=0$) but nonzero magnetic moment ($\mu'\neq0$) subject to a static electric field ($\partial_t\mathbf{E}=0$, $\partial_t\phi=0$) but with no magnetic field ($\mathbf{B}=0$, $\mathbf{A}=0$).\footnote{A Dirac-Pauli spinor with $q=0$ but $\mu'\neq0$ can be used to describe spin-$1/2$ uncharged baryons such as protons. However, this description only gives an effective theory as Pauli's prescription for inclusion of anomalous magnetic moment is only phenomenological.} The condition \eqnref{condition calX} now reads as
\begin{equation}\label{condition calX case II}
2mc^2\calX=-\calX\Omega\calX+\Omega^\dag,
\end{equation}
where we define the operators
\begin{subequations}
\begin{eqnarray}
\Omega &:=& c\,\boldsigma\cdot\mathbf{p} +i\mu'\sdotE,\\
\Omega^\dag &:=& c\,\boldsigma\cdot\mathbf{p} -i\mu'\sdotE.
\end{eqnarray}
\end{subequations}
Multiplying $\Omega$ on \eqnref{condition calX case II} from the left yields a quadratic equation in $\Omega\calX$:
\begin{equation}
\left(\Omega\calX\right)^2+2mc^2\left(\Omega\calX\right)-\Omega\Omega^\dag = 0.
\end{equation}
This admits an exact solution
\begin{equation}\label{OmegaX}
\Omega\calX = -mc^2 +\sqrt{m^2c^4+\Omega\Omega^\dag},
\end{equation}
which is manifestly hermitian, i.e.,
\begin{equation}\label{OmegaX dag}
\Omega\calX =\left(\Omega\calX\right)^\dag\equiv\calX^\dag\Omega^\dag.
\end{equation}
Meanwhile, multiplying $\calX^\dag$ on \eqnref{condition calX case II} from the left and applying \eqnref{OmegaX dag}, we have
\begin{equation}
\left(2mc^2+\Omega\calX\right)\calX^\dag\calX=\calX^\dag\Omega^\dag=\Omega\calX,
\end{equation}
which follows
\begin{equation}
\calX^\dag\calX=\frac{\Omega\calX}{2mc^2+\Omega\calX}.
\end{equation}
As $\calX^\dag\calX$ is a function of $\Omega\calX$,  $\calX^\dag\calX$ commutes with $\Omega\calX$. As a result, \eqnref{calHFW Kutzelnigg} gives
\begin{eqnarray}\label{HFW case II}
\calH_\mathrm{FW}
&=& mc^2 + \sqrt{1+\calX^\dag\calX}
\left(\Omega\calX\right)
\frac{1}{\sqrt{1+\calX^\dag\calX}}\nonumber\\
&=&mc^2+\Omega\calX
=\sqrt{m^2c^4+\Omega\Omega^\dag}\nonumber\\
&=&\Big[m^2c^4 +c^2\mathbf{p}^2 -\mu'\hbar\,c\,\boldsymbol{\nabla}\cdot\mathbf{E} +\mu'c\left(\mathbf{p}\times\mathbf{E}-\mathbf{E}\times\mathbf{p}\right)\cdot\boldsigma
+\mu'^2\mathbf{E}^2\Big]^{1/2},
\end{eqnarray}
where we have used \eqnref{identity 1} and \eqnref{identity 3} to compute $\Omega\Omega^\dag$.

Like the first special case, the resulting FW transformed Hamiltonian in \eqnref{HFW case II} is exactly the same as that obtained by Eriksen's method \cite{Eriksen1958,Chen2014}.
Unlike the first special case, however, the physical interpretation and relevance of the fact that the Hamiltonian can be exactly block-diagonalized is not well understood, as the second special case is rather artificial. Closer investigations into the mathematical structure of QED for further insight are needed.

\subsection{Weak-field limit}\label{sec:weak-field limit}
When the external electromagnetic field is weak enough, we expect that the FW transformed Hamiltonian exists and agrees with the classical Hamiltonian given by \eqnref{H cl}--\eqnref{H spin} except for some quantum corrections that have no classical correspondence. By denoting the Dirac or Dirac-Pauli Hamiltonian as $\tilde{\calH}(\phi,\mathbf{A},\mathbf{E},\mathbf{B})$, the rigorous mathematical statement reads as follows. The $4\times4$ unitary matrix $\tilde{U}$ exists such that the \emph{formal} linear-field limit defined as
\begin{equation}\label{weak-field limit}
\lim_{\lambda\rightarrow0}
\frac{\tilde{U}\
\tilde{\calH}(\lambda\phi,\lambda\mathbf{A},\lambda\mathbf{E},\lambda\mathbf{B})
\tilde{U}^\dagger}
{\lambda}
\end{equation}
is block-diagonal and in agreement with the classical counterpart, even though $\tilde{\calH}$ itself might not be exactly diagonalizable. Physically, this means the particle-antiparticle separation remains legitimate when the electromagnetic field is weak enough so that the energy interacting with electromagnetic fields does not exceed the Dirac energy gap. It should be noted that while the FW transformed Hamiltonian is only \emph{approximate} from the physical point of view, it is \emph{exact} in the formal limit \eqnref{weak-field limit} from the mathematical point of view.

As detailed in \cite{Chen2014}, the two special cases in Sections \ref{sec:special case I} and \ref{sec:special case II} in conjunction suggest that, in the weak-field limit, the FW transformed Dirac-Pauli Hamiltonian takes the form
\begin{eqnarray}\label{calHFW conjecture}
\calH_\mathrm{FW}(\mathbf{x},\mathbf{p},\boldsigma)
&=&\sqrt{c^2\boldsymbol{\pi}^2+m^2c^4}\,+q\phi\nonumber\\
&&
\mbox{}
-\frac{\hbar}{2}\,\boldsigma\cdot
\Bigg[
\left(\gamma'_m+\frac{q}{mc}\frac{1}{\gamma_{\boldpi}}\right)\mathbf{B}
-\gamma'_m\frac{1}{\gamma_{\boldpi}(1+\gamma_{\boldpi})}
\frac{\overline{\left(\boldpi\cdot\mathbf{B}\right)\boldpi}}{m^2c^2}
\nonumber\\
&&\qquad\qquad
\mbox{}-\left(\gamma'_m\frac{1}{\gamma_{\boldpi}} +\frac{q}{mc}\frac{1}{\gamma_{\boldpi}(1+\gamma_{\boldpi})}\right)
\frac{\overline{\boldpi\times\mathbf{E}}}{mc}
\Bigg]_\mathrm{Weyl}
\nonumber\\
&&\mbox{}
+\frac{\hbar^2}{4mc}\left(\frac{q}{2mc}-\gamma'_m\right)
\left(\frac{\boldsymbol{\nabla}\cdot\mathbf{E}}{\gamma_{\boldsymbol{\pi}}}
\right)_\mathrm{Weyl},
\end{eqnarray}
where $\overline{(\cdots)}$ and $(\cdots)_\mathrm{Weyl}$ denote specific symmetrization for operator orderings defined in \cite{Chen2014}.
$\calH_\mathrm{FW}$ in \eqnref{calHFW conjecture} is in full agreement with the classical counterpart given by \eqnref{H cl}--\eqnref{H spin} with $\mathbf{s}=\hbar\boldsigma/2$ except for the operator orderings and the Darwin term involving $\hbar^2$, both of which have no classical correspondence.

The form of \eqnref{calHFW conjecture} is conjectured from the two special cases, which are complementary to each other, and still requires further confirmation for the cases in the presence of both $\mathbf{E}$ and $\mathbf{B}$. Its validity has been confirmed in \cite{Chen2013} by Kutzelnigg's method of DPT up to the order of $(\frac{\boldpi}{mc})^{14}$ for the case of static and homogeneous electromagnetic fields, whereby the Darwin term vanishes and there are no complications arising from operator orderings thanks to homogeneity, and the FW transformation remains explicitly time-independent and thus in conformity with the standard FW scenario  thanks to staticity \cite{Chen2014}.  Based on the results obtained in \cite{Chen2013}, we are able to prove by mathematical induction that, in static and homogeneous electromagnetic fields, the FW transformed Hamiltonian in the weak-field limit is completely in agreement with the classical counterpart. We present the proof first for the Dirac Hamiltonian in \secref{sec:Dirac Hamiltonian} and then for the Dirac-Pauli Hamiltonian in \secref{sec:Dirac-Pauli Hamiltonian}.

\subsection{Remarks on the FW transformation}\label{sec:remarks}
The main purpose of this paper is to prove the correspondence between classical and Dirac-Pauli spinors via the FW transformation. We do not intend to settle the disputed issues about the mathematical rigor and legitimacy of the FW transformation but only briefly remark on some of them.

First of all, it should be emphasized again that, for generic settings, the Dirac equation is not self-consistent without second quantization (i.e., quantization in quantum field theory). The inconsistency can be seen from the fact that the Dirac equation gives rise to the Klein paradox (as the Klein-Gordon equation does), rendering the first quantization formalism non-unitary (see Section 5.6 of \cite{Strange2008} for more details). This implies that the \emph{exact} FW transformation does not exist except for some special settings (such as the special cases presented above and those discussed in \cite{Nikitin:1998}), or otherwise it would exactly decouple the particle from the antiparticle and thus remove the Klein paradox without appealing to second quantization. Apart from some special conditions that admit the exact FW transformation, the FW transformation exists \emph{exactly} only in some \emph{formal limit} (i.e., when some \emph{regularization} is properly prescribed) such as the weak-field limit prescribed in \eqnref{weak-field limit}.

In the literature of relativistic quantum mechanics, many \emph{exact}-decoupling methods of the FW transformation have been constructed and used for various applications (e.g., see \cite{Reiher2004a,Reiher2004b}). Rigorously speaking, exactness of these methods should be understood in the sense that some regularization has been prescribed although usually the prescription is not explicitly specified and might seem obscure. That said, existence of the exact FW transformation is often taken for granted before a method is formulated, and only when the method is used in actual applications is some regularization then tacitly prescribed. For example, in the work of \cite{Reiher2004b}, when the Douglas-Kroll-Hess method \cite{Douglas1974,Hess1986} is applied to one-electron atoms, calculations have been performed with an even-tempered universal Gaussian basis set, the employment of which can be viewed as a prescription of regularization imposed to suppress infinitely long-range effects of the Coulomb potential. (Also see \cite{Peng2012} for more discussions on other theoretical aspects of exact-decoupling methods.)

The FW methods can be classified into two types: the one-step (direct) approach and the order-by-order (step-by-step) approach (see \cite{Silenko2013-analysis} for a comparative analysis of these two approaches). Many methods give a closed form of the one-step solution but the closed form so obtained usually remains formal (see \cite{Peng2012} for more comments) except for some special cases (as presented above). In order to reveal the relevant physics, one has to adopt an order-by-order approach in the first place or to further perform order-by-order expansion upon the one-step solution. In the order-by-order approach, it is crucial to know whether the power series converges or not. The issue of convergence has been carefully investigated in \cite{Reiher2004a,Reiher2004b} (also see \cite{Reiher2009 Book} for a detailed review). In the series expansion in terms of $1/c$, the radius of convergence (in the complex plane of momentum space) is finite. In this regard, the expansion in $1/c$ is deemed inadequate on the grounds that it is divergent for large momenta. On the other hand, the series expansion in terms of the scalar potential $\phi$, known as the Douglas-Kroll-Hess method \cite{Douglas1974,Hess1986}, is convergent on a sliced complex plane of momentum space that covers the \emph{whole} real axis. Therefore, the Douglas-Kroll-Hess method is adequate for any value of momenta.

It should be noted that the aforementioned pathology of the expansion in $1/c$ simply means that, at some point when the momentum is large enough, it will stop being a good approximation to the exact FW transformed Hamiltonian if the series expansion is truncated to a \emph{finite} series. This, however, does \emph{not} invalidate the closed-form solution obtained from the \emph{infinite} series \emph{as a whole}. If the whole infinite series converges to a closed form of an analytic function within the radius of convergence, the analytic function can then be extended beyond the radius of convergence via analytic continuation.\footnote{For example, $(1-z)^{-1}$ admits the power series $\sum_{n=0}^\infty z^n$ for $|z|<1$. This does not imply that $(1-z)^{-1}$ is well defined only for $|z|<1$; on the contrary, it is well defined and analytic everywhere in the complex plane except $z=1$.}$^,$\footnote{Also see the last paragraph in \secref{sec:calHFW}, especially \eqnref{Gaussian int}, for a formal implementation of the analytic continuation used for our proof of the quantum-classical correspondence.} Therefore, as long as the closed form of the infinite power series is attainable, the order-by-order method in terms of $1/c$ is as valid as the Douglas-Kroll-Hess method and, furthermore, the closed-form solutions are unique (more precisely, unitarily equivalent to one another) whatever approaches are taken (provided they are regularized in equivalent ways). This is exactly what happens in the rest of this paper for the proof of the correspondence between classical and Dirac-Pauli spinors.\footnote{We could have used the Douglas-Kroll-Hess method for our purpose if it is accordingly modified to incorporate the vector potential $\mathbf{A}$ in addition to the scalar potential $\phi$. If the modification is formulated in a fashion that the series expansion is in terms of $\phi$ and $\mathbf{A}$, then our desired linear-field limit can be readily obtained as the first-order result. However, this modification does not seem straightforward at all. Furthermore, even in the ordinary Douglas-Kroll-Hess method (i.e., in the absence of $\mathbf{A}$), the first-order result cannot be directly compared to the conjectured form \eqnref{calHFW conjecture}, but further series expansion has to be performed. It turns out the Douglas-Kroll-Hess method is less suitable for our purpose and instead we adopt Kutzelnigg's method of DPT.}

To sum up, despite some doubts about the legitimacy of the FW transformation in general and of the approach we adopt in particular, our proof remains sound on account of the two facts: first, regularization is properly prescribed for the weak-field limit as in \eqnref{weak-field limit}; second, the exact solution is obtained in a closed form as in \eqnref{calHFW}.

\section{Dirac Hamiltonian}\label{sec:Dirac Hamiltonian}
For the Dirac theory, we first solve the operator $X$ by the power series expansion and then obtain the FW transformed Hamiltonian $H_\mathrm{FW}$. As we assume the applied electromagnetic fields to be static and homogeneous, we have $[\pi_i,E_j]=[\pi_i,B_j]=0$. Moreover, because we focus on the weak-field limit, we neglect all the terms nonlinear in $F_{\mu\nu}$.

\subsection{Operators $X_n$}
The operator $X$ used in Kutzelnigg's method of DPT satisfies the condition \eqnref{condition X} for the Dirac theory.
Consider the power series of $X$ in powers of $c^{-1}$:
\begin{equation}
X=\sum_{j=1}^\infty\frac{X_j}{c^j}.
\end{equation}
For the orders of $1/c$ and $1/c^2$, \eqnref{condition X} yields
\begin{subequations}
\begin{eqnarray}
2mX_1&=&\sdotp,\\
\label{X2}
2mX_{2}&=&0.
\end{eqnarray}
\end{subequations}
According to \eqnref{condition X}, the higher-order terms in the power series of $X$ can be determined by the following recursion relations (for $j\geq1$):
\begin{subequations}\label{recursion Xj}
\begin{eqnarray}
\label{recursion a}
2mX_{2j}&=&-\sum_{k_1+k_2=2j-1}X_{k_1}\sdotp X_{k_2}+q[\phi,X_{2j-2}],\\
\label{recursion b}
2mX_{2j+1}&=&-\sum_{k_1+k_2=2j}X_{k_1}\sdotp X_{k_2}+q[\phi,X_{2j-1}].
\end{eqnarray}
\end{subequations}
Explicitly, the leading terms $X_j$ read as
\begin{subequations}\label{leading Xj}
\begin{eqnarray}
\label{X1}
X_1&=&\frac{\sdotp}{2m},\\
\label{X3}
X_3&=&-\frac{1}{8}\frac{(\sdotp)^3}{m^3} -\frac{1}{4}\frac{iq\hbar}{m^2}\,\sdotE,\\
X_5&=&\frac{1}{16}\frac{(\sdotp)^5}{m^5} +\frac{3}{16}\frac{iq\hbar}{m^4}\,\boldpi^2(\sdotE) +\frac{1}{8}\frac{iq\hbar}{m^4}\,(\sdotp)(\Edotp),\\
X_7&=&-\frac{5}{128}\frac{(\sdotp)^7}{m^7} -\frac{5}{32}\frac{iq\hbar}{m^6}\,\boldpi^4(\sdotE) -\frac{3}{16}\frac{iq\hbar}{m^6}\,\boldpi^2(\sdotp)(\Edotp),\\
X_9&=&\frac{7}{256}\frac{(\sdotp)^9}{m^9} +\frac{35}{256}\frac{iq\hbar}{m^8}\,\boldpi^6(\sdotE) +\frac{29}{128}\frac{iq\hbar}{m^8}\,\boldpi^4(\sdotp)(\Edotp),\\
X_{11}&=&-\frac{21}{1024}\frac{(\sdotp)^{11}}{m^{11}} -\frac{63}{1024}\frac{iq\hbar}{m^{10}}\,\boldpi^8(\sdotE) -\frac{65}{256}\frac{iq\hbar}{m^{10}}\boldpi^6\,(\sdotp)(\Edotp),\\
X_{13}&=&\frac{33}{2048}\frac{(\sdotp)^{13}}{m^{13}} +\frac{231}{2048}\frac{iq\hbar}{m^{12}}\,\boldpi^{10}(\sdotE) +\frac{281}{1024}\frac{iq\hbar}{m^{12}}\,\boldpi^{8}(\sdotp)(\Edotp),
\end{eqnarray}
\end{subequations}
and $X_{2j}=0$ for all $j$. (These were laboriously calculated in \cite{Chen2013}.)

Based on the result of \eqnref{leading Xj}, we can conjecture the following theorem and provide its proof by mathematical induction.
\begin{theorem}\label{thm:1}
In the weak-field limit, we neglect nonlinear terms in $\mathbf{E}$ and $\mathbf{B}$. If the electromagnetic field is homogeneous (thus, $[\pi_i,E_j]=[\pi_i,B_j]=0$), the generic expression for $X_{n\geq0}$ is given by
\begin{subequations}\label{X2j and X2j+1}
\begin{eqnarray}
\label{X2j}
X_{2j}&=&0,\\
\label{X2j+1}
X_{2j+1}&=&a_j\frac{(-1)^j}{(2m)^{2j+1}}(\sdotp)^{2j+1}
+b_j\frac{iq\hbar(-1)^j}{(2m)^{2j}}\,\boldpi^{2j-2}(\sdotE)\nonumber\\
&&\mbox{}
+c_j\frac{iq\hbar(-1)^j}{(2m)^{2j}}\,\boldpi^{2j-4}(\sdotp)(\Edotp),
\end{eqnarray}
\end{subequations}
where the coefficients are defined as
\begin{subequations}
\begin{eqnarray}
\label{aj}
a_{j\geq0} &=& \frac{(2j)!}{j!(j+1)!},\\
\label{bj}
b_{j\ge1} &=& \frac{(2j-1)!}{j!(j-1)!} \equiv (2j-1)a_{j-1}, \qquad b_{j=0}=0,\\
\label{cj}
c_{j\geq0}&=&2\sum_{j_1+j_2=j}b_{j_1}b_{j_2}, \qquad (\text{particularly},\ c_{j=0,1}=0).
\end{eqnarray}
\end{subequations}
\end{theorem}
\begin{proof}[Proof (by induction)]
It is trivial to prove \eqnref{X2j} by applying \eqnref{recursion a} on \eqnref{X2} inductively.
To prove \eqnref{X2j+1}, we first note that it is valid for $j=1$ by \eqnref{X3}. Suppose \eqnref{X2j+1} is true for all $X_{2k+1}$ with $k<j$. Since $X_{2k}=0$, the recursive relation \eqnref{recursion b} reads as
\begin{equation}
2mX_{2j+1} = -\sum_{j_1+j_2=j-1}X_{2j_1+1}(\sdotp)X_{2j_2+1} +q[\phi,\,X_{2j-1}],
\end{equation}
which, by applying the inductive hypothesis for $k<j$, yields
\begin{subequations}\label{pf 1}
\begin{eqnarray}
\label{pf 1a}
2mX_{2j+1} &=& -\sum_{j_1+j_2=j-1}X_{2j_1+1}(\sdotp)X_{2j_2+1} +q\left[\phi,\,a_{j-1}\frac{(-1)^{j-1}}{(2m)^{2j-1}}(\sdotp)^{2j-1}\right]\\
&=& -\sum_{j_1+j_2=j-1}a_{j_1}a_{j_2}\frac{(-1)^{j_1+j_2}}{(2m)^{2(j_1+j_2)+2}}(\sdotp)^{2(j_1+j_2)+3}\nonumber\\
\label{pf 1b}
&&\mbox{}
-2iq\hbar\sum_{j_1+j_2=j-1}a_{j_1}b_{j_2}\frac{(-1)^{j_1+j_2}}{(2m)^{2(j_1+j_2)+1}} (\sdotp)^{2j_1+2}\boldpi^{2j_2-2}(\sdotE)\nonumber\\
&&\mbox{}
-2iq\hbar\sum_{j_1+j_2=j-1}a_{j_1}c_{j_2}\frac{(-1)^{j_1+j_2}}{(2m)^{2(j_1+j_2)+1}} (\sdotp)^{2j_1+2}\boldpi^{2j_2-4}(\sdotp)(\Edotp)\nonumber\\
&&\mbox{}
+q\,a_{j-1}\frac{(-1)^{j-1}}{(2m)^{2j-1}}\left[\phi,\,(\sdotp)^{2j-1}\right],
\end{eqnarray}
\end{subequations}
where in \eqnref{pf 1a} we have neglected nonlinear terms in $\mathbf{E}$ and in \eqnref{pf 1b} adopted $[\pi_i,E_j]=0$.
Next, applying \eqnref{identity 3} and \eqnref{identity 4b} and dropping out the second term in \eqnref{identity 3} whenever it is accompanied by $\mathbf{E}$, we then have
\begin{eqnarray}
2mX_{2j+1} &=&
-\sum_{j_1+j_2=j-1}a_{j_1}a_{j_2}\frac{(-1)^{j_1+j_2}}{(2m)^{2(j_1+j_2)+2}}(\sdotp)^{2(j_1+j_2)+3}\nonumber\\
&&\mbox{}
-2iq\hbar\sum_{j_1+j_2=j-1}a_{j_1}b_{j_2}\frac{(-1)^{j_1+j_2}}{(2m)^{2(j_1+j_2)+1}} \boldpi^{2(j_1+j_2)}(\sdotE)\nonumber\\
&&\mbox{}
-2iq\hbar\sum_{j_1+j_2=j-1}a_{j_1}c_{j_2}\frac{(-1)^{j_1+j_2}}{(2m)^{2(j_1+j_2)+1}} \boldpi^{2(j_1+j_2)-2}(\sdotp)(\Edotp)\nonumber\\
&&\mbox{}
-iq\hbar\,a_{j-1}\frac{(-1)^{j-1}}{(2m)^{2j-1}}\,\boldpi^{2j-2}(\sdotE)\nonumber\\
&&\mbox{}
-2iq\hbar\,a_{j-1}(j-1)\frac{(-1)^{j-1}}{(2m)^{2j-1}}\,\boldpi^{2j-4}(\sdotp)(\Edotp).
\end{eqnarray}
Consequently, we have
\begin{eqnarray}
X_{2j+1} &=&
\left(\,\sum_{j_1+j_2=j-1}a_{j_1}a_{j_2}\right)\frac{(-1)^{j}}{(2m)^{2j+1}}(\sdotp)^{2j+1}\nonumber\\
&&\mbox{}
+iq\hbar\left(2\sum_{j_1+j_2=j-1}a_{j_1}b_{j_2}+a_{j-1}\right)\frac{(-1)^{j}}{(2m)^{2j}} \boldpi^{2j-2}(\sdotE)\nonumber\\
&&\mbox{}
+iq\hbar\left(2\sum_{j_1+j_2=j-1}a_{j_1}c_{j_2}+2(j-1)a_{j-1}\right)\frac{(-1)^{j}}{(2m)^{2j}} \boldpi^{2j-4}(\sdotp)(\Edotp),
\end{eqnarray}
which can be shown to take the form of \eqnref{X2j+1} by the combinatorial identities (their proofs will be provided shortly):
\begin{subequations}\label{combo identities}
\begin{eqnarray}
\label{combo identity a}
\text{for}\ j\geq1:\qquad
\sum_{j_1+j_2=j-1}a_{j_1}a_{j_2}&=&a_j,\\
\label{combo identity b}
2\sum_{j_1+j_2=j-1}a_{j_1}b_{j_2}&=&b_j-a_{j-1}\equiv2(j-1)a_{j-1},\\
\label{combo identity c}
2\sum_{j_1+j_2=j-1}a_{j_1}c_{j_2}&\equiv&4\sum_{j_1+j_2+j_3=j-1}a_{j_1}b_{j_2}b_{j_3}\nonumber\\
=c_j-b_j+a_j&\equiv& c_j-2(j-1)a_{j-1}.
\end{eqnarray}
\end{subequations}
We therefore have proved the theorem by mathematical induction.
\end{proof}

\subsection{Operators $X$ and $X^\dag$}
We have the Taylor series with the radius of convergence $|x|<1$:
\begin{subequations}\label{series abc}
\begin{eqnarray}
\label{series a}
\sum_{j=0}^\infty a_j\frac{(-1)^j}{2^{2j+1}}\,x^{2j+1}
&=&
\frac{x}{1+\sqrt{1+x^2}}\equiv x^{-1}\left(\sqrt{1+x^2}-1\right),\\
\label{series b}
\sum_{j=1}^\infty b_j\frac{(-1)^j}{2^{2j}}\,x^{2j-2}
&=&
\frac{1}{2}\left(\frac{1}{1+\sqrt{1+x^2}}-\frac{1}{\sqrt{1+x^2}}\right),\\
\label{series c}
\sum_{j=2}^\infty c_j\frac{(-1)^j}{2^{2j}}\,x^{2j-4}
&=&
\frac{1}{8}\left(\frac{1}{1+\sqrt{1+x^2}}-\frac{1}{\sqrt{1+x^2}}\right)^2,
\end{eqnarray}
\end{subequations}
where \eqnref{series a} and \eqnref{series b} are obtained by the binomial series: $(1+x)^{\pm1/2}=\sum_{n=0}^\infty{\pm1/2 \choose n}x^n$.\footnote{Conversely, we have
\begin{subequations}
\begin{eqnarray}
\label{binomial series a}
\sqrt{1+x^2}&=&1+\sum_{j=0}^\infty a_j\frac{(-1)^jx^{2(j+1)}}{2^{2j+1}},\\
\label{binomial series b}
\frac{1}{\sqrt{1+x^2}}&=&\sum_{j=0}^\infty (a_j+b_{j+1})\frac{(-1)^jx^{2j}}{2^{2j+1}}
=\sum_{j=0}^\infty (j+1)a_j\frac{(-1)^jx^{2j}}{2^{2j}}.
\end{eqnarray}
\end{subequations}}
Meanwhile, with $c_j$ defined by \eqnref{cj}, taking squares on both sides of \eqnref{series b} immediately yields \eqnref{series c}.

The combinatorial identities \eqnref{combo identities} can be proven by the above Taylor series. Taking squares on both sides of \eqnref{series a} gives
\begin{eqnarray}
&& \sum_{j_1,j_2=0}^\infty a_{j_1}a_{j_2}\frac{(-1)^{j_1+j_2}}{2^{2(j_1+j_2)+2}}\,x^{2(j_1+j_2)+1}
=\sum_{j=0}^\infty\sum_{j_1+j_2=j}a_{j_1}a_{j_2}\frac{(-1)^j}{2^{2j+2}}\,x^{2j+1}\nonumber\\
&=&\sum_{j=1}^\infty\sum_{j_1+j_2=j-1}a_{j_1}a_{j_2}\frac{(-1)^{j-1}}{2^{2j+2}}\,x^{2j+2}
=1-\frac{2}{x}\sum_{j=0}^\infty\sum_{j_1+j_2=j-1}a_{j_1}a_{j_2}\frac{(-1)^j}{2^{2j+1}}\,x^{2j+1}\nonumber\\
&=&\left(\frac{x}{1+\sqrt{1+x^2}}\right)^2,
\end{eqnarray}
which leads to
\begin{eqnarray}
&&\sum_{j=0}^\infty\sum_{j_1+j_2=j-1}a_{j_1}a_{j_2}\frac{(-1)^j}{2^{2j+1}}\,x^{2j+1}\nonumber\\
&=&-\frac{x}{2}\left(\left(\frac{x}{1+\sqrt{1+x^2}}\right)^2-1\right)
=\frac{x}{1+\sqrt{1+x^2}}.
\end{eqnarray}
By \eqnref{series a} again, we obtain \eqnref{combo identity a}. The identity \eqnref{combo identity b} can be proved similarly, and \eqnref{combo identity c} follows immediately from \eqnref{combo identity b} with the definition \eqnref{cj}.
Additionally, exploiting \eqnref{series abc} in a similar way enables us to prove one more combinatorial identity:
\begin{equation}\label{combo identity d}
\text{for}\ j\geq0:\quad b_{j+1}+c_{j+1}=4b_j+4c_j+a_j,
\end{equation}
which will be useful later.

By \eqnref{X2j and X2j+1}, we obtain the Taylor series of the $X$ operator:
\begin{eqnarray}\label{series X}
X&=&\sum_{k=1}^\infty\frac{X_k}{c^k}=\sum_{j=0}^\infty\frac{X^{2j+1}}{c^{2j+1}}\nonumber\\
&=&\sum_{j=0}^\infty a_j\frac{(-1)^j}{(2mc)^{2j+1}}\,(\sdotp)^{2j+1}
+\frac{iq\hbar}{c}\sum_{j=1}^\infty b_j\frac{(-1)^j}{(2mc)^{2j}}\,\boldpi^{2j-2}(\sdotE)\nonumber\\
&&\mbox{}+\frac{iq\hbar}{c}\sum_{j=2}^\infty c_j\frac{(-1)^j}{(2mc)^{2j}}\,\boldpi^{2j-4} (\sdotp)(\Edotp).
\end{eqnarray}
Adopting $[\pi_i,E_j]=0$, we have
\begin{eqnarray}\label{series Xdagger}
X^\dag&=&
\sum_{j=0}^\infty a_j\frac{(-1)^j}{(2mc)^{2j+1}}\,(\sdotp)^{2j+1}
-\frac{iq\hbar}{c}\sum_{j=1}^\infty b_j\frac{(-1)^j}{(2mc)^{2j}}\,\boldpi^{2j-2}(\sdotE)\nonumber\\
&&\mbox{}-\frac{iq\hbar}{c}\sum_{j=2}^\infty c_j\frac{(-1)^j}{(2mc)^{2j}}\,\boldpi^{2j-4} (\sdotp)(\Edotp).
\end{eqnarray}
By \eqnref{series abc}, the Taylor series of the operator $X$ given in \eqnref{series X} converges to a closed form provided that
\begin{equation}\label{convergence condition}
\left|(\sdotp)^2\right|=\left|\boldpi^2-\frac{q\hbar}{c}\sdotB\right|<m^2c^2.
\end{equation}
We will discuss the condition for convergence in the end of \secref{sec:calHFW}.

Adopting $[\pi_i,E_j]=0$ again and neglecting nonlinear terms in $\mathbf{E}$, \eqnref{series X} and \eqnref{series Xdagger} then give
\begin{eqnarray}\label{XdaggerX}
X^\dag X &=& \sum_{j_1,j_2=0}^\infty a_{j_1}a_{j_2}\frac{(-1)^{j+j_2}}{(2mc)^{2(j_1+j_2)+2}}\,(\sdotp)^{2(j_1+j_2)+2}\nonumber\\
&&\mbox{}
+\frac{iq\hbar}{c}\sum_{j_1=0,j_2=1}^\infty a_{j_1}b_{j_2}\frac{(-1)^{j_1+j_2}}{(2mc)^{2(j_1+j_2)+1}}\,\boldpi^{2(j_1+j_2)-2} \left[\sdotp,\sdotE\right]\nonumber\\
&=&\sum_{j=0}^\infty\sum_{j_1+j_2=j}a_{j_1}a_{j_2}
\frac{(-1)^{j}}{(2mc)^{2j+2}}\,(\sdotp)^{2j+2}\nonumber\\
&&\mbox{}
+2\,\frac{q\hbar}{c}\sum_{j=1}^\infty\sum_{j_1+j_2=j}^\infty a_{j_1}b_{j_2}\frac{(-1)^{j}}{(2mc)^{2j+1}}\,\boldpi^{2j-2} (\Etimesp)\cdot\boldsigma\nonumber\\
&=& \sum_{j=0}^\infty a_{j+1}\frac{(-1)^j}{(2mc)^{2j+2}}\,(\sdotp)^{2j+2}\nonumber\\
&&\mbox{}
+\frac{q\hbar}{c}\sum_{j=1}^\infty (b_{j+1}-a_j) \frac{(-1)^{j}}{(2mc)^{2j+1}}\,\boldpi^{2j-2} (\Etimesp)\cdot\boldsigma,
\end{eqnarray}
where \eqnref{identity 1}, \eqnref{identity 3}, and \eqnref{combo identities} have been used.

\subsection{Operator $H_\mathrm{FW}$}

Before we calculate $H_\mathrm{FW}$, let us investigate the operators $[q\phi,(X^\dag X)]$ and $[c\,\sdotp X,(X^\dag X)^n]$ beforehand.

First, by \eqnref{XdaggerX} and \eqnref{identity 4a}, we have
\begin{eqnarray}\label{qphi XdaggerX}
[q\phi,X^\dag X] &=&
\sum_{j=0}^\infty a_{j+1}\frac{(-1)^j}{(2mc)^{2j+2}}
\left[q\phi,(\sdotp)^{2j+2}\right]\nonumber\\
&=&iq\hbar\sum_{j=0}^\infty 2(j+1)a_{j+1}\frac{(-1)^j}{(2mc)^{2j+2}}\,
\boldpi^{2j}(\Edotp).
\end{eqnarray}
Note that $[X^\dag X,\boldpi^{2j}(\Edotp)]=0$ if we neglect nonlinear terms in $F_{\mu\nu}$ and adopt $[\pi_i,E_j]=0$. Consequently, by induction, we have
\begin{equation}
[q\phi,(X^\dag X)^n]
=n[q\phi,X^\dag X](X^\dag X)^{n-1},
\end{equation}
for $n\geq1$. Expanding $(1+x)^{1/2} = \sum_{n=0}^\infty {1/2 \choose n}x^n \equiv \sum_{n=0}^\infty e_n x^n$, we can then compute
\begin{eqnarray}\label{HFW lemma 1}
&&\sqrt{1+X^\dag X}\,(q\phi) \equiv \sum_{n=0}^\infty e_n (X^\dag X)^n(q\phi)\nonumber\\
&=& \sum_{n=0}^\infty e_n (q\phi)(X^\dag X)^n -\sum_{n=1}^\infty ne_n[q\phi,X^\dag X](X^\dag X)^{n-1}\nonumber\\
&=& (q\phi)\sqrt{1+X^\dag X}\, -[q\phi,X^\dag X]\frac{1}{2\sqrt{(1+X^\dag X)}},
\end{eqnarray}
where we have used $\frac{d}{dx}(1+x)^{1/2} =\frac{1}{2}(1+x)^{-1/2} =\sum_{n=1}^\infty ne_n x^{n-1}$.

Second, from \eqnref{series X}, we get
\begin{eqnarray}\label{h0X}
c(\sdotp)X
&=&c\sum_{j=0}^\infty a_j\frac{(-1)^j}{(2mc)^{2j+1}}\,(\sdotp)^{2j+2}
+q\hbar\sum_{j=1}^\infty b_j\frac{(-1)^j}{(2mc)^{2j}}\,\boldpi^{2j-2} (\Etimesp)\cdot\boldsigma\nonumber\\
&&\mbox{}+iq\hbar\sum_{j=1}^\infty (b_j+c_j)\frac{(-1)^j}{(2mc)^{2j}}\,\boldpi^{2j-2} (\Edotp),
\end{eqnarray}
where \eqnref{identity 1} and \eqnref{identity 3} have been used and the superfluous term involving $c_{j=1}=0$ is added for bookkeeping convenience.
Note that, up to the linear terms in $F_{\mu\nu}$, the $\sdotB$ piece of \eqnref{identity 3} can be dropped out for the factors $(\sdotp)^{2j+2}$ in both \eqnref{XdaggerX} and \eqnref{h0X} when we compute $\left[c\,\sdotp X, X^\dag X\right]$. Consequently we have
\begin{equation}\label{HFW lemma 2}
\left[c\,\sdotp X, X^\dag X\right] =0.
\end{equation}

We are now ready to calculate $H_\mathrm{FW}$. With \eqnref{HFW lemma 1} and \eqnref{HFW lemma 2}, \eqnref{HFW Kutzelnigg} leads to
\begin{subequations}\label{HFW 1}
\begin{eqnarray}
\label{HFW 1a}
H_\mathrm{FW} &=& mc^2 + \sqrt{1+X^\dag X}\,(q\phi+c\,\sdotp X)\frac{1}{\sqrt{1+X^\dag X}}\\
&=& mc^2+q\phi -[q\phi,X^\dag X]\frac{1}{2(1+X^\dag X)} +c\,\sdotp X.
\end{eqnarray}
\end{subequations}
Substituting \eqnref{qphi XdaggerX} and \eqnref{h0X} into \eqnref{HFW 1} gives
\begin{eqnarray}\label{HFW 2}
H_\mathrm{FW} &=& mc^2 + q\phi +c\sum_{j=0}^\infty a_j\frac{(-1)^j}{(2mc)^{2j+1}}\,(\sdotp)^{2j+2}
+q\hbar\sum_{j=1}^\infty b_j\frac{(-1)^j}{(2mc)^{2j}}\,\boldpi^{2j-2} (\Etimesp)\cdot\boldsigma\nonumber\\
&&\mbox{}+iq\hbar\sum_{j=1}^\infty (b_j+c_j)\frac{(-1)^j}{(2mc)^{2j}}\,\boldpi^{2j-2} (\Edotp)\nonumber\\
&&\mbox{}
-iq\hbar\left(\sum_{j=0}^\infty (j+1)a_{j+1}\frac{(-1)^j}{(2mc)^{2j+2}}\,
\boldpi^{2j}(\Edotp)\right)\frac{1}{1+X^\dag X}.
\end{eqnarray}
Because $H_\mathrm{FW}$ is hermitian, the last two terms in \eqnref{HFW 2}, which give the anti-hermitian part, are expected to cancel each other exactly. This can be seen explicitly by checking vanishing of the following composition of operators:
\begin{eqnarray}\label{HFW antihermitian}
&&\left(\sum_{j=1}^\infty (b_j+c_j)\frac{(-1)^j}{(2mc)^{2j}}\,\boldpi^{2j-2}\right)(1+X^\dag X)
+\sum_{j=0}^\infty (j+1)a_{j+1}\frac{(-1)^j}{(2mc)^{2j+2}}\,\boldpi^{2j}\nonumber\\
\label{check a}
&=&\left(\sum_{j=1}^\infty (b_j+c_j)\frac{(-1)^j}{(2mc)^{2j}}\,\boldpi^{2j-2}\right)
\left(1+\sum_{j=0}^\infty a_{j+1}\frac{(-1)^j}{(2mc)^{2j+2}}\,\boldpi^{2j+2}\right)\nonumber\\
&&\quad
-\sum_{j=1}^\infty j a_j \frac{(-1)^j}{(2mc)^{2j}}\,\boldpi^{2j-2}\nonumber\\
&=&\sum_{j=1}^\infty (b_j+c_j)\frac{(-1)^j}{(2mc)^{2j}}\,\boldpi^{2j-2}
+\sum_{j_1,j_2=1}^\infty (a_{j_1}b_{j_2}+a_{j_1}c_{j_2}) \frac{(-1)^{j_1+j_2+1}}{(2mc)^{2(j_1+j_2)}}\,\boldpi^{2(j_1+j_2)-2}\nonumber\\
&&\quad
-\sum_{j=1}^\infty j a_j \frac{(-1)^j}{(2mc)^{2j}}\,\boldpi^{2j-2}\nonumber\\
&=&\sum_{j=1}^\infty (b_j+c_j)\frac{(-1)^j}{(2mc)^{2j}}\,\boldpi^{2j-2}
+\sum_{j=2}^\infty\sum_{j_1+j_2=j \atop j_1,j_2\neq0}^\infty (a_{j_1}b_{j_2}+a_{j_1}c_{j_2}) \frac{(-1)^{j+1}}{(2mc)^{2j}}\,\boldpi^{2j-2}\nonumber\\
&&\quad
-\sum_{j=1}^\infty j a_j \frac{(-1)^j}{(2mc)^{2j}}\,\boldpi^{2j-2}\nonumber\\
&=&\frac{a_1-b_1-c_1}{(2mc)^2}
+\sum_{j=2}^\infty \left(b_j+c_j-ja_j -\sum_{j_1+j_2=j \atop j_1,j_2\neq0}^\infty (a_{j_1}b_{j_2}+a_{j_1}c_{j_2})\right)\frac{(-1)^j}{(2mc)^{2j}}\,\boldpi^{2j-2},
\end{eqnarray}
where in the second line we have dropped out the $\sdotB$ piece of \eqnref{identity 3} for the factors $(\sdotp)^{2j+2}$ in \eqnref{XdaggerX}.
For each coefficient factor of the summand, we have
\begin{eqnarray}
&&b_j+c_j-ja_j -\sum_{j_1+j_2=j \atop j_1,j_2\neq0}^\infty (a_{j_1}b_{j_2}+a_{j_1}c_{j_2})
\equiv2b_j+2c_j-ja_j -\sum_{j_1+j_2=j}^\infty (a_{j_1}b_{j_2}+a_{j_1}c_{j_2})\nonumber\\
&=&2b_j+2c_j-\frac{1}{2}\left(b_{j+1}+c_{j+1}-a_j\right)
\end{eqnarray}
by \eqnref{combo identities}, and it vanishes identically by \eqnref{combo identity d}. Also note that $a_1-b_1-c_1=0$. We thus show that \eqnref{HFW antihermitian} vanishes, thereby affirming hermiticity of $H_\mathrm{FW}$.

As the antihermitian part vanishes, \eqnref{HFW 2} leads to
\begin{eqnarray}\label{HFW 3}
H_\mathrm{FW}
 &=& mc^2 + q\phi +c\sum_{j=0}^\infty a_j\frac{(-1)^j}{(2mc)^{2j+1}}\, (\sdotp)^{2j+2}
 +q\hbar\sum_{j=1}^\infty b_j\frac{(-1)^j}{(2mc)^{2j}}\,\boldpi^{2j-2}
 (\Etimesp)\cdot\boldsigma\nonumber\\
&=&mc^2 + q\phi +mc^2\sum_{j=0}^\infty a_j\frac{(-1)^j}{2^{2j+1}}\left(\frac{\sdotp}{mc}\right)^{2j+2}
+\frac{q\hbar}{(mc)^2} \sum_{j=1}^\infty b_j\frac{(-1)^j}{2^{2j}}\left(\frac{\boldpi}{mc}\right)^{2j-2} (\Etimesp)\cdot\boldsigma\nonumber\\
&=& mc^2+q\phi +mc^2\left(\sqrt{1+\left(\frac{\sdotp}{mc}\right)^2}-1\right)\nonumber\\
&&\mbox{}
+\frac{q\hbar}{2(mc)^2} \left(\frac{1}{1+\sqrt{1+\left(\frac{\boldpi}{mc}\right)^2}}
-\frac{1}{\sqrt{1+\left(\frac{\boldpi}{mc}\right)^2}}\right)
\boldsigma\cdot(\Etimesp),
\end{eqnarray}
where the Taylor series \eqnref{series a} and \eqnref{series b} are used. Note that, up to the linear order in $\mathbf{B}$, we have
\begin{eqnarray}
&&\sqrt{1+\left(\frac{\sdotp}{mc}\right)^2}
=\sqrt{1+\left(\frac{\boldpi}{mc}\right)^2-\frac{q\hbar}{m^2c^3}\sdotB}\nonumber\\
&=&\sqrt{1+\left(\frac{\boldpi}{mc}\right)^2} \left(1-\frac{1}{2}\frac{q\hbar}{m^2c^3} \frac{\sdotB}{\left(1+\left(\frac{\boldpi}{mc}\right)^2\right)}+\cdots\right).
\end{eqnarray}
Taking this back into \eqnref{HFW 3}, we obtain
\begin{eqnarray}\label{HFW 4}
H_\mathrm{FW} &=& q\phi +\sqrt{m^2c^4+c^2\boldpi^2} -\frac{q\hbar}{2mc}\frac{1}{\gamma_{\boldpi}}\sdotB\nonumber\\
&&\quad
\mbox{}+\frac{q\hbar}{2mc}\left(\frac{1}{\gamma_{\boldpi}}-\frac{1}{1+\gamma_{\boldpi}}\right) \boldsigma\cdot\left(\frac{\boldpi}{mc}\times\mathbf{E}\right),
\end{eqnarray}
where the Lorentz factor associated with the kinematic momentum $\boldpi$ is defined as
\begin{equation}\label{gamma pi}
\gamma_{\boldpi}:=\sqrt{1+\left(\frac{\boldpi}{mc}\right)^2}
\equiv \sum_{n=0}^\infty {1/2 \choose n}\left(\frac{\boldpi}{mc}\right)^{2n}
\end{equation}
in accordance with the classical counterpart \eqnref{gamma pi cl}.
The FW transform of the Dirac Hamiltonian given in \eqnref{HFW 4} fully agrees with the classical counterpart \eqnref{H cl}--\eqnref{H spin} with $\mathbf{s}=\frac{\hbar}{2}\boldsigma$ and $\gamma'_m=0$ (or $\gamma_m=\frac{q}{mc}$).

\section{Dirac-Pauli Hamiltonian}\label{sec:Dirac-Pauli Hamiltonian}
As we have proved the exact correspondence between the Dirac Hamiltonian and the classical counterpart in the weak-field limit, we now extend the result to the Dirac-Pauli theory. Again, we first solve the operator $\calX$ by the power series expansion and then obtain the FW transformed Hamiltonian $\calH_\mathrm{FW}$. We again assume $[\pi_i,E_j]=[\pi_i,B_j]=0$ for homogeneous fields and neglect all the terms nonlinear in $F_{\mu\nu}$ in the weak-field limit.

\subsection{Operators $X'_n$}
For the Dirac-Pauli theory, the operator $\calX$ used in Kutzelnigg's method satisfies the condition \eqnref{condition calX}, which reads as
\begin{eqnarray}
{2mc^2}\calX&=&-\calX c\,\sdotp\,\calX+c\,\sdotp+q[\phi,\calX]\nonumber\\
&& -i\frac{\mu''}{c}\sdotE-i\frac{\mu''}{c}\calX\sdotE\,\calX +\frac{\mu''}{c}\{\calX,\sdotB\},
\end{eqnarray}
where we define
\begin{equation}
\mu'':=c\mu',
\end{equation}
as it is more convenient to factor out the dimensionality of $c^{-1}$ in $\mu'$ for the power series method in powers of $c^{-1}$.

Consider the power series of $\calX$ in powers of $c^{-1}$:
\begin{equation}
\calX:=X+X'=\sum_{j=1}^\infty\frac{\calX_j}{c^j}
=\sum_{j=1}^\infty\frac{X_j}{c^j}
+\sum_{j=1}^{\infty}\frac{X'_j}{c^j},
\end{equation}
where $X$ and $X_j$ have been detailed in \secref{sec:Dirac Hamiltonian}.
For the orders of $1/c$, $1/c^2$ and $1/c^3$, we have
\begin{subequations}
\begin{align}
2m\calX_1&=\sdotp, & &\Rightarrow X_1=\eqnref{X1}, & X'_1&=0,\\
2m\calX_2&=0, & &\Rightarrow  X_2=0, & X'_2&=0\\
2m\calX_3&= -\calX_1\sdotp\calX_1 +q[\phi,\calX_1] -i\mu''\sdotE,
& &\Rightarrow  X_3=\eqnref{X3}, & X'_3&=-\frac{i\mu''}{2m}\sdotE.
\end{align}
\end{subequations}
The higher-order terms in the power series of $\calX$ can be determined by the following recursion relations ($j\geq2$):
\begin{subequations}\label{recursion calXj}
\begin{eqnarray}
2m\calX_{2j}&=&-\sum_{k_1+k_2=2j-1}\calX_{k_1}\sdotp\,\calX_{k_2} +q\left[\phi,\calX_{2j-2}\right]\nonumber\\
&&\mbox{}-i\mu''\sum_{k_1+k_2=2j-3}\calX_{k_1}\sdotE\,\calX_{k_2} +\mu''\left\{\calX_{2j-3},\sdotB\right\}\\
2m\calX_{2j+1}&=&-\sum_{k_1+k_2=2j}\calX_{k_1}\sdotp\,\calX_{k_2} +q\left[\phi,\calX_{2j-1}\right]\nonumber\\
&&\mbox{}-i\mu''\sum_{k_1+k_2=2j-2}\calX_{k_1}\sdotE\,\calX_{k_2} +\mu''\left\{\calX_{2j-2},\sdotB\right\},
\end{eqnarray}
\end{subequations}
which together with \eqnref{recursion Xj} lead to the recursion relation for $X'_n$ ($j\geq2$):
\begin{subequations}\label{recursion X'j}
\begin{eqnarray}
\label{recursion a'}
2mX'_{2j}&=&-\sum_{k_1+k_2=2j-1} \left(X_{k_1}\sdotp X'_{k_2}+X'_{k_1}\sdotp X_{k_2}+X'_{k_1}\sdotp X'_{k_2}\right)\nonumber\\
&&\mbox{}-i\mu''\sum_{k_1+k_2=2j-3}\left(X_{k_1}\sdotE\,X_{k_2}+X_{k_1}\sdotE\,X'_{k_2}+X'_{k_1}\sdotE\, X_{k_2}+X'_{k_1}\sdotp X'_{k_2}\right)\nonumber\\
&&\mbox{} +q\left[\phi,X'_{2j-2}\right] +\mu''\left\{X_{2j-3}+X'_{2j-3},\sdotB\right\},\\
\label{recursion b'}
2mX'_{2j+1}&=&-\sum_{k_1+k_2=2j} \left(X_{k_1}\sdotp X'_{k_2}+X'_{k_1}\sdotp X_{k_2}+X'_{k_1}\sdotp X'_{k_2}\right)\nonumber\\
&&\mbox{}-i\mu''\sum_{k_1+k_2=2j-2}\left(X_{k_1}\sdotE\,X_{k_2}+X_{k_1}\sdotE\,X'_{k_2}+X'_{k_1}\sdotE\, X_{k_2}+X'_{k_1}\sdotp X'_{k_2}\right)\nonumber\\
&&\mbox{} +q\left[\phi,X'_{2j-1}\right] +\mu''\left\{X_{2j-2}+X'_{2j-2},\sdotB\right\}.
\end{eqnarray}
\end{subequations}
Neglecting nonlinear terms in $\mathbf{E}$ and $\mathbf{B}$, the leading terms $X'_j$ read as
\begin{subequations}\label{leading X'j}
\begin{align}
  X'_1&=0, & X'_2&=0, \\
  X'_3&= -\frac{i\mu''}{2m}\sdotE, & X'_4&=\frac{\mu''}{2m^2}\Bdotp,\\
  X'_5&=\frac{3}{8}\frac{i\mu''}{m^3}\boldpi^2(\sdotE)-\frac{i\mu''}{4m^3}(\sdotp)(\Edotp),
  &
  X'_6&=-\frac{3}{8}\frac{\mu''}{m^4}\boldpi^2(\Bdotp),\\
  X'_7&=-\frac{5}{16}\frac{i\mu''}{m^5}\boldpi^4(\sdotE) +\frac{1}{4}\frac{i\mu''}{m^5}\boldpi^2(\sdotp)(\Edotp),
  &
  X'_8&=\frac{5}{16}\frac{\mu''}{m^6}\boldpi^4(\Bdotp),\\
  X'_9&=\frac{35}{128}\frac{i\mu''}{m^7}\boldpi^6(\sdotE)-\frac{15}{64}\frac{i\mu''}{m^7}\boldpi^4(\sdotp)(\Edotp),
  &
  X'_{10}&=-\frac{35}{128}\frac{\mu''}{m^8}\boldpi^6(\Bdotp),\\
  X'_{11}&=-\frac{63}{256}\frac{i\mu''}{m^9}\boldpi^8(\sdotE) +\frac{7}{32}\frac{i\mu''}{m^9}\boldpi^6(\sdotp)(\Edotp),
  &
  X'_{12}&=\frac{63}{256}\frac{\mu''}{m^{10}}\boldpi^8(\Bdotp).
\end{align}
\end{subequations}
(These where laboriously calculated in \cite{Chen2013}.)

Based on the result of \eqnref{leading X'j}, we can conjecture the following theorem and provide its proof by mathematical induction.
\begin{theorem}\label{thm:2}
In the weak-field limit, we neglect nonlinear terms in $\mathbf{E}$ and $\mathbf{B}$. If the electromagnetic field is homogeneous (thus, $[\pi_i,E_j]=[\pi_i,B_j]=0$), the generic expression for $X'_{n\geq2}$ is given by
\begin{subequations}\label{X'2j and X'2j+1}
\begin{eqnarray}
\label{X'2j}
X'_{2j}&=&2b_{j-1}\frac{\mu''(-1)^{j}}{(2m)^{2j-2}}\,\boldpi^{2j-4}(\Bdotp),\\
\label{X'2j+1}
X'_{2j+1}&=&b_j\frac{i\mu''(-1)^{j}}{(2m)^{2j-1}}\,\boldpi^{2j-2}(\sdotE)\nonumber\\
&&\mbox{}
+d_j\frac{i\mu''(-1)^{j+1}}{(2m)^{2j-1}}\,\boldpi^{2j-4}(\sdotp)(\Edotp),
\end{eqnarray}
\end{subequations}
where the coefficients $b_j$ are given by \eqnref{bj} and $d_j$ are defined as
\begin{equation}\label{dj}
d_{j\geq2}=\sum_{j_1+j_2+j_3=j-2}2(j_1+1)a_{j_1}a_{j_2}a_{j_3},
\qquad
d_{j=0}=d_{j=1}=0.
\end{equation}
\end{theorem}
\begin{proof}[Proof (by induction)]
Note that \eqnref{X'2j and X'2j+1} is valid for $j=1$ and $j=2$ by \eqnref{leading X'j}. Suppose \eqnref{X'2j and X'2j+1} is true for all $X_{2k}$ and $X_{2k+1}$ with $k<j$, we will prove $X'_{2j}$ and $X'_{2j+1}$ to be true for $j\geq2$ by induction.

First, we prove \eqnref{X'2j} for $j\geq2$. With the inductive hypothesis and \eqnref{X2j and X2j+1}, the recursive relation \eqnref{recursion a'} yields
\begin{equation}
2mX'_{2j} = -\sum_{j_1+j_2=j-1}\left(X_{2j_1+1}(\sdotp)X'_{2j_2}+X'_{2j_2}(\sdotp)X_{2j_1+1}\right) +\mu''\left\{X_{2j-3},\,\sdotB\right\},
\end{equation}
where we have neglected nonlinear terms in $\mathbf{E}$ and $\mathbf{B}$.
Applying the inductive hypothesis for $k<j$ and \eqnref{X2j+1}, we have
\begin{eqnarray}\label{pf 2a}
&&2mX'_{2j}\nonumber\\
&=& -\mu''\sum_{j_1+j_2=j-1}2a_{j_1}b_{j_2-1}\frac{(-1)^{j_1+j_2}}{(2m)^{2(j_1+j_2)-1}} \Big((\sdotp)^{2j_1+2}\boldpi^{2j_2-4}(\Bdotp) +\boldpi^{2j_2-4}(\Bdotp)(\sdotp)^{2j_1+2}\Big)\nonumber\\
&&\mbox{}+\mu''a_{j-2}\,\frac{(-1)^{j-2}}{(2m)^{2j-3}} \Big((\sdotp)^{2j-3}(\sdotB)+(\sdotB)(\sdotp)^{2j-3}\Big)\nonumber\\
&=&-2\mu''\sum_{j_1+j_2=j-2}2a_{j_1}b_{j_2}\frac{(-1)^{j_1+j_2+1}}{(2m)^{2(j_1+j_2)+1}} \,\boldpi^{2(j_1+j_2)}(\Bdotp)\nonumber\\
&&\mbox{}+\mu''a_{j-2}\,\frac{(-1)^{j-2}}{(2m)^{2j-3}}\, (\sdotp)^{2j-4}
\Big((\sdotp)(\sdotB)+(\sdotB)(\sdotp)\Big)(\sdotp)^{2j-4}\nonumber\\
&=&2\mu''\left(2\sum_{j_1+j_2=j-2}a_{j_1}b_{j_2}+a_{j-2}\right)
\frac{(-1)^{j}}{(2m)^{2j-3}}\,\boldpi^{2j-4}(\Bdotp),
\end{eqnarray}
where we have used \eqnref{identity 3} to throw away nonlinear terms in $\mathbf{B}$ and used \eqnref{identity 1} with $[\pi_i,B_j]=0$ to get
\begin{eqnarray}
&&(\sdotp)(\sdotB)+(\sdotB)(\sdotp)\nonumber\\
&=&\boldpi\cdot\mathbf{B}+\Bdotp +i\left(\boldpi\times\mathbf{B}+\mathbf{B}\times\boldpi\right)\cdot\boldsigma\nonumber\\
&=&2(\Bdotp).
\end{eqnarray}
By the combinatorial identity \eqnref{combo identity b}, it follows from \eqnref{pf 2a} that $X'_{2j}$ for $j\geq2$ takes the form of \eqnref{X'2j}.

Next, we prove \eqnref{X'2j+1} for $j\geq2$. With the inductive hypothesis and \eqnref{X2j and X2j+1} again, the recursive relation \eqnref{recursion b'} yields
\begin{eqnarray}
2mX'_{2j+1} &=& -\sum_{j_1+j_2=j-1}\left(X_{2j_1+1}(\sdotp)X'_{2j_2+1}+X'_{2j_2+1}(\sdotp)X_{2j_1+1}\right)\nonumber\\ &&\mbox{}-i\mu''\sum_{j_1+j_2=j-2}X_{2j+1}(\sdotE)X_{2j+1},
\end{eqnarray}
where we have neglected nonlinear terms in $\mathbf{E}$ and $\mathbf{B}$.
Applying the inductive hypothesis for $k<j$ and \eqnref{X2j+1}, we have
\begin{eqnarray}\label{pf 2b}
2mX'_{2j+1}
&=&\mbox{}
-i\mu''\sum_{j_1+j_2=j-1} a_{j_1}b_{j_2}\frac{(-1)^{j_1+j_2}}{(2m)^{2(j_1+j_2)}} \,(\sdotp)^{2j_1+2}\boldpi^{2j_2-2}(\sdotE)\nonumber\\
&&\mbox{}
-i\mu''\sum_{j_1+j_2=j-1} a_{j_1}d_{j_2}\frac{(-1)^{j_1+j_2+1}}{(2m)^{2(j_1+j_2)}} \,(\sdotp)^{2j_1+2}\boldpi^{2j_2-4}(\sdotp)(\sdotE)\nonumber\\
&&\mbox{}
-i\mu''\sum_{j_1+j_2=j-1} a_{j_1}b_{j_2}\frac{(-1)^{j_1+j_2}}{(2m)^{2(j_1+j_2)}} \,\boldpi^{2j_2-2}(\sdotE)(\sdotp)(\sdotp)^{2j_1+1}\nonumber\\
&&\mbox{}
-i\mu''\sum_{j_1+j_2=j-1} a_{j_1}d_{j_2}\frac{(-1)^{j_1+j_2+1}}{(2m)^{2(j_1+j_2)}} \,\boldpi^{2j_2-4}(\sdotp)(\Edotp)(\sdotp)(\sdotp)^{2j_1+1}\nonumber\\
&&\mbox{}
-i\mu''\sum_{j_1+j_2=j-2} a_{j_1}a_{j_2}\frac{(-1)^{j_1+j_2}}{(2m)^{2(j_1+j_2)+2}} \,(\sdotp)^{2j_1+1}(\sdotE)(\sdotp)^{2j_2+1}.
\end{eqnarray}
By using \eqnref{identity 3} to throw away nonlinear terms in $\mathbf{B}$ and using \eqnref{identity 1} with $[\pi_i,B_j]=0$ to get
\begin{eqnarray}
&&(\sdotp)(\sdotE)(\sdotp)\nonumber\\
&=&\Big((\boldpi\cdot\mathbf{E})+i(\boldpi\times\mathbf{E})\cdot\boldsigma\Big)(\sdotp)\nonumber\\
&=&(\boldpi\cdot\mathbf{E})(\sdotp) +i(\boldpi\times\mathbf{E})\cdot\boldpi -\left((\boldpi\times\mathbf{E})\times\boldpi\right)\cdot\boldsigma\nonumber\\
&=&(\boldpi\cdot\mathbf{E})(\sdotp) +\left((\boldpi\cdot\mathbf{E})\boldpi-\boldpi^2\mathbf{E}\right)\cdot\boldsigma\nonumber\\
&=&2(\sdotp)(\Edotp)-\boldpi^2(\sdotE),
\end{eqnarray}
\eqnref{pf 2b} then leads to
\begin{eqnarray}\label{pf 2b'}
2mX'_{2j+1}
&=&\mbox{}
-i\mu''\sum_{j_1+j_2=j-1} a_{j_1}b_{j_2}\frac{(-1)^{j_1+j_2}}{(2m)^{2(j_1+j_2)}} \,\boldpi^{2(j_1+j_2)}(\sdotE)\nonumber\\
&&\mbox{}
-i\mu''\sum_{j_1+j_2=j-1} a_{j_1}d_{j_2}\frac{(-1)^{j_1+j_2+1}}{(2m)^{2(j_1+j_2)}} \,\boldpi^{2(j_1+j_2)-2}(\sdotp)(\sdotE)\nonumber\\
&&\mbox{}
-i\mu''\sum_{j_1+j_2=j-1} a_{j_1}b_{j_2}\frac{(-1)^{j_1+j_2}}{(2m)^{2(j_1+j_2)}} \,\boldpi^{2(j_1+j_2)}(\sdotE)\nonumber\\
&&\mbox{}
-i\mu''\sum_{j_1+j_2=j-1} a_{j_1}d_{j_2}\frac{(-1)^{j_1+j_2+1}}{(2m)^{2(j_1+j_2)}} \,\boldpi^{2(j_1+j_2)-2}(\sdotp)(\Edotp)\nonumber\\
&&\mbox{}
-2i\mu''\sum_{j_1+j_2=j-2} a_{j_1}a_{j_2}\frac{(-1)^{j_1+j_2}}{(2m)^{2(j_1+j_2)+2}} \,\boldpi^{2(j_1+j_2)}(\sdotp)(\Edotp)\nonumber\\
&&\mbox{}
+i\mu''\sum_{j_1+j_2=j-2} a_{j_1}a_{j_2}\frac{(-1)^{j_1+j_2}}{(2m)^{2(j_1+j_2)+2}} \,\boldpi^{2(j_1+j_2)+2}(\sdotE),
\end{eqnarray}
and consequently
\begin{eqnarray}\label{pf 2b''}
X'_{2j+1}
&=&
i\mu''\left(\sum_{j_1+j_2=j-1} 2a_{j_1}b_{j_2} +\sum_{j_1+j_2=j-1} a_{j_1}a_{j_2}\right)\frac{(-1)^{j}}{(2m)^{2j-2}} \,\boldpi^{2j-2}(\sdotE)\\
&&\mbox{}
+2i\mu''\left(\sum_{j_1+j_2=j-1} a_{j_1}d_{j_2} +\sum_{j_1+j_2=j-1} a_{j_1}a_{j_2}\right) \frac{(-1)^{j+1}}{(2m)^{2j-2}} \,\boldpi^{2(j_1+j_2)-2}(\sdotp)(\sdotE).\nonumber
\end{eqnarray}
The combinatorial identities \eqnref{combo identity a} and \eqnref{combo identity b} immediately imply that the summations inside the first pair of parentheses in \eqnref{pf 2b''} are equal to $b_j$. Furthermore, by \eqnref{combo identity a} and the new combinatorial identity (its proof will be provided shortly)
\begin{equation}\label{combo identity e}
\text{for}\ j\geq2: \quad 2\sum_{j_1+j_2=j-1} a_{j_1}d_{j_2} + 2a_{j-1} = d_j,
\end{equation}
the summations inside the second pair of parentheses in \eqnref{pf 2b''} are equal to $d_j$. Consequently, it follows from \eqnref{pf 2b''} that $X'_{2j+1}$ for $j\geq2$ takes the form of \eqnref{X'2j+1}.

We have proved both \eqnref{X'2j} and \eqnref{X'2j+1} by mathematical induction.
\end{proof}

\subsection{Operators $X'$ and $X'^\dag$}
We have the Taylor series with the radius of convergence $|x|<1$:
\begin{equation}
\label{series d}
\sum_{j=2}^\infty d_j\frac{(-1)^{j}}{2^{2j-1}}\,x^{2j-4}
=\frac{1}{\sqrt{1+x^2}} \left(\frac{1}{1+\sqrt{1+x^2}}\right)^2,
\end{equation}
which, with $d_j$ defined by \eqnref{dj}, can be proven by taking squares on both sides of \eqnref{series a} and then multiplying both sides by \eqnref{binomial series b}.
Similarly, exploiting \eqnref{series abc} and \eqnref{series d} also enables us to prove the combinatorial identities \eqnref{combo identity e} and
\begin{equation}\label{combo identity f}
\text{for}\ j\geq0: \quad b_{j+1}+a_j=d_{j+1}.
\end{equation}

By \eqnref{X'2j and X'2j+1}, we obtain the Taylor series of the $X'$ operator:
\begin{eqnarray}\label{series X'}
X' &=& \sum_{j=1}^\infty\frac{X'_j}{c^j}= \sum_{j=1}^\infty\frac{X'_{2j}}{c^{2j}} + \sum_{j=1}^\infty\frac{X'_{2j+1}}{c^{2j+1}}\nonumber\\
&=& -2\mu''\sum_{j=1}^\infty b_j\frac{(-1)^j}{(2mc)^{2j}}\,\boldpi^{2j-2}(\Bdotp)
+i\mu''\sum_{j=1}^\infty b_j\frac{(-1)^j}{(2mc)^{2j-1}}\,\boldpi^{2j-2}(\sdotE)\nonumber\\
&&\mbox{} -i\mu''\sum_{j=2}^\infty d_j\frac{(-1)^j}{(2mc)^{2j-1}}\,\boldpi^{2j-4}(\sdotp)(\Edotp).
\end{eqnarray}
Adopting $[\pi_i,E_j]=[\pi_i,B_j]=0$, we have
\begin{eqnarray}\label{series X'dagger}
X'^\dag &=&
-2\mu''\sum_{j=1}^\infty b_j\frac{(-1)^j}{(2mc)^{2j}}\,\boldpi^{2j-2}(\Bdotp)
-i\mu''\sum_{j=1}^\infty b_j\frac{(-1)^j}{(2mc)^{2j-1}}\,\boldpi^{2j-2}(\sdotE)\nonumber\\
&&\mbox{} +i\mu''\sum_{j=2}^\infty d_j\frac{(-1)^j}{(2mc)^{2j-1}}\,\boldpi^{2j-4}(\sdotp)(\Edotp).
\end{eqnarray}
By \eqnref{series b} and \eqnref{series d}, the Taylor series of the operator $X'$ given in \eqnref{series X'} converges to a closed form provided that
\begin{equation}\label{convergence condition'}
\left|\boldpi^2\right|<m^2c^2.
\end{equation}
We will discuss the condition for convergence in the end of \secref{sec:calHFW}.

\subsection{Operator $\calH_\mathrm{FW}$}\label{sec:calHFW}
We have \eqnref{calHFW Kutzelnigg}
with
\begin{equation}
\calX=X+X'.
\end{equation}
Because $X'$ is of the order $\bigO(F_{\mu\nu})$ as shown in \eqnref{series X'}, up to $\bigO(F_{\mu\nu})$, \eqnref{calHFW Kutzelnigg} leads to
\begin{eqnarray}\label{HFW+H'FW}
\calH_\mathrm{FW}
&=& mc^2 + \sqrt{1+X^\dag X}
\left(q\phi +c\,\sdotp X\right)
\frac{1}{\sqrt{1+X^\dag X}}\nonumber\\
&&\mbox{}
+\left(c\,\sdotp X' -\mu'\sdotB +i\mu'\sdotE\,X\right)\nonumber\\
&=:&H_\mathrm{FW}+H'_\mathrm{FW},
\end{eqnarray}
where the first half part is identified as $H_\mathrm{FW}$ by \eqnref{HFW 1a}, and the second half is called $H'_\mathrm{FW}$.

By \eqnref{series X} and \eqnref{series X'}, we have
\begin{eqnarray}\label{H'FW 1}
H'_\mathrm{FW} &=& c\,\sdotp X' -\mu'\sdotB +i\mu'\sdotE\,X \nonumber\\
&=& -2\mu'\sum_{j=1}^\infty b_j\frac{(-1)^j}{(2mc)^{2j}}\,\boldpi^{2j-2}(\sdotp)(\Bdotp)\nonumber\\
&&\mbox{} +i\mu'\sum_{j=1}^\infty b_j\frac{(-1)^j}{(2mc)^{2j-1}}\,\boldpi^{2j-2}(\Edotp)
+\mu'\sum_{j=1}^\infty b_j\frac{(-1)^j}{(2mc)^{2j-1}}\,\boldpi^{2j-2}(\Etimesp)\cdot\boldsigma\nonumber\\
&&\mbox{}-i\mu'\sum_{j=2}^\infty d_j\frac{(-1)^j}{(2mc)^{2j-1}}\,\boldpi^{2j-2}(\Edotp)\nonumber\\
&&\mbox{}-\mu'\sdotB\nonumber\\
&&\mbox{} -i\mu'\sum_{j=0}^\infty a_j\frac{(-1)^j}{(2mc)^{2j+1}}\,\boldpi^{2j}(\Edotp)
-\mu'\sum_{j=0}^\infty a_j\frac{(-1)^j}{(2mc)^{2j+1}}\,\boldpi^{2j}(\Etimesp)\cdot\boldsigma,
\end{eqnarray}
where we have used \eqnref{identity 1} and \eqnref{identity 3} and neglected nonlinear terms in $F_{\mu\nu}$. Equation~\eqnref{H'FW 1} leads to
\begin{eqnarray}\label{H'FW 2}
H'_\mathrm{FW}
&=& -2\mu'\sum_{j=1}^\infty b_j\frac{(-1)^j}{(2mc)^{2j}}\,\boldpi^{2j-2}(\sdotp)(\Bdotp)\nonumber\\
&&\mbox{} +\mu'\left(\sum_{j=1}^\infty b_j\frac{(-1)^j}{(2mc)^{2j-1}}\,\boldpi^{2j-2}
-\sum_{j=0}^\infty a_j\frac{(-1)^j}{(2mc)^{2j+1}}\,\boldpi^{2j}\right)
(\Etimesp)\cdot\boldsigma\nonumber\\
&&\mbox{}-\mu'\sdotB\nonumber\\
&&\mbox{} -i\mu'\sum_{j=0}^\infty \left(b_{j+1}-d_{j+1}+a_j\right) \frac{(-1)^j}{(2mc)^{2j+1}}\,\boldpi^{2j}(\Edotp).
\end{eqnarray}
By \eqnref{combo identity f}, we find that the antihermitian part in \eqnref{H'FW 2} vanishes identically. Furthermore, by \eqnref{series a} and \eqnref{series b}, we have
\begin{eqnarray}\label{H'FW 3}
H'_\mathrm{FW}
&=& -\mu'\left(\frac{1}{1+\sqrt{1+\left(\frac{\boldpi}{mc}\right)^2}} -\frac{1}{\sqrt{1+\left(\frac{\boldpi}{mc}\right)^2}}\right)
\frac{(\sdotp)(\Bdotp)}{(mc)^2}\nonumber\\
&&\mbox{} -\mu'\left(\frac{1}{\sqrt{1+\left(\frac{\boldpi}{mc}\right)^2}}\right)
\frac{(\Etimesp)\cdot\boldsigma}{mc}
-\mu'\sdotB\nonumber\\
&=&\mu'\left(\frac{1}{\gamma_{\boldpi}}-\frac{1}{1+\gamma_{\boldpi}}\right) \boldsigma\cdot\frac{\boldpi}{mc}\left(\frac{\boldpi}{mc}\cdot\mathbf{B}\right)
+\mu'\frac{1}{\gamma_{\boldpi}}\boldsigma\cdot\left(\frac{\boldpi}{mc}\times\mathbf{E}\right)
-\mu'\sdotB,
\end{eqnarray}
where $\gamma_{\boldpi}$ is defined in \eqnref{gamma pi}.

With \eqnref{HFW 4} and \eqnref{H'FW 3}, we have
\begin{eqnarray}\label{calHFW}
\calH_\mathrm{FW}(\mathbf{x},\mathbf{p},\boldsigma) &=& H_\mathrm{FW}+H'_\mathrm{FW}\nonumber\\
&=&\sqrt{m^2c^4+c^2\boldpi^2}\,+q\phi(\mathbf{x})\nonumber\\
&&\mbox{}-\boldsigma\cdot\left[
\left(\mu'+\frac{q\hbar}{2mc}\frac{1}{\gamma_{\boldpi}}\right)\mathbf{B}
-\mu'\frac{1}{\gamma_{\boldpi}(1+\gamma_{\boldpi})}
\left(\frac{\boldpi}{mc}\cdot\mathbf{B}\right)\frac{\boldpi}{mc}
\right.\nonumber\\
&&\qquad\qquad
\left.
\mbox{}-\left(\mu'\frac{1}{\gamma_{\boldpi}} +\frac{q\hbar}{2mc}\frac{1}{\gamma_{\boldpi}(1+\gamma_{\boldpi})}\right)
\left(\frac{\boldpi}{mc}\times\mathbf{E}\right)
\right],
\end{eqnarray}
which is exactly the same as \eqnref{calHFW conjecture} except that the Darwin term vanishes and the operator orderings are superfluous. This proves that, in the weak-field limit, the FW transform of the Dirac-Pauli Hamiltonian is in complete agreement with the classical counterpart \eqnref{H cl}--\eqnref{H spin} with $\mathbf{s}=\frac{\hbar}{2}\boldsigma$ and $\mu'=\frac{\hbar}{2}\gamma'_m$.

Note that, by applying the Taylor series \eqnref{series abc} and \eqnref{series d}, the functions of the operator $\Omega=\sdotp/(mc)$ or $\Omega=\boldpi/(mc)$ are understood via the Taylor series as
\begin{equation}
f\left(1+\Omega^2\right)
=\sum_{n=0}^{\infty}\frac{f^{(n)}(1)}{n!}\,\Omega^{2n},
\end{equation}
which produces convergent results provided that the spectrum of $\Omega$ satisfies $|\Omega^2|<1$. This requires the conditions of \eqnref{convergence condition} and \eqnref{convergence condition'} to be satisfied.
In comparison with the classical theory in the weak-field regime, in which $\boldpi$ remains as the kinematic momentum associated with $\mathbf{v}$ as indicated by \eqnref{pi appro} and \eqnref{gamma pi appro}, the conditions \eqnref{convergence condition} and \eqnref{convergence condition'} correspond to $|\mathbf{v}|<c/\sqrt{2}$ (which is well beyond the low-speed limit).
Once the operators $X$ and $X'$ converge to closed forms for $|\mathbf{v}|<c/\sqrt{2}$, their closed forms are in fact upheld even beyond the conditions of \eqnref{convergence condition} and \eqnref{convergence condition'}. This is because, instead of the Taylor series \eqnref{series abc} and \eqnref{series d}, the pertinent function $1/\sqrt{1+\Omega^2}$ can be alternatively understood in terms of the integral
\begin{equation}\label{Gaussian int}
\frac{1}{\sqrt{1+\Omega^2}} =\lim_{N\rightarrow\infty} \int_{-N}^N d\eta\, e^{-\pi\eta^2(1+\Omega^2)},
\end{equation}
where the exponential operator is defined by means of its Taylor expansion. The form of \eqnref{Gaussian int} gives convergent results for \emph{all} $\Omega$.\footnote{Here, we have adopted the idea propounded in \cite{Eriksen1958}.}
Therefore, even though the Taylor series \eqnref{series abc} and \eqnref{series d} break down when \eqnref{convergence condition} and \eqnref{convergence condition'} do not hold, the resulting $\calH_\mathrm{FW}$ in \eqnref{calHFW} as a closed form nevertheless remains valid (as long as the applied electromagnetic field is weak enough so that nonlinear terms in $F_{\mu\nu}$ can be neglected).

\section{Summary and discussion}\label{sec:summary}
In Kutzelnigg's implementation of DPT improved with a further simplification scheme, the FW transform of the Dirac-Pauli Hamiltonian is given by \eqnref{calHFW Kutzelnigg} with $\calX$ satisfying \eqnref{condition calX}, which reduces to \eqnref{HFW Kutzelnigg} with $X$ satisfying \eqnref{condition X} for the Dirac Hamiltonian. For the two special cases studied in \secref{sec:special case I} and \secref{sec:special case II}, the exact FW transformed Hamiltonians exist and agree with those obtained by Eriksen's method \cite{Eriksen1958}. Existence of the exact FW transformation in the first special case is accordant with the fact that charged particle-antiparticle pairs are not produced by any static magnetic field no matter how strong the field strength is \cite{Kim:2000un,Kim:2011cx}. On the other hand, the physical relevance of the exact FW transformation in the second case is unclear and requires further research.

The conditions for the operators $X$ and $\calX\equiv X+X'$ give rise to the recursion relations \eqnref{recursion Xj}, \eqnref{recursion calXj}, and \eqnref{recursion X'j} for their power series. When the applied electromagnetic field is static and homogeneous, in the weak-field limit in which nonlinear terms in $F_{\mu\nu}$ are neglected, we have \thmref{thm:1} and \thmref{thm:2}, which are proven by mathematical induction via the recursion relations and various combinatorial identities. Consequently, the resulting FW transformed Dirac-Pauli Hamiltonian in the weak-field limit is given by \eqnref{calHFW}, which is in full agreement with the classical counterpart \eqnref{H cl}--\eqnref{H spin} with $\mathbf{s}=\frac{\hbar}{2}\boldsigma$ and $\mu'=\frac{\hbar}{2}\gamma'_m$.

If the applied electromagnetic field is inhomogeneous, it is suggested in \cite{Chen2014} that the FW transform in the weak-field limit takes the form of \eqnref{calHFW conjecture}, which is an extension of \eqnref{calHFW} with corrections of the Darwin term and operator orderings. A rigorous proof of \eqnref{calHFW conjecture} in the style of this paper is however much more difficult, as it is very cumbersome to keep track of operator orderings in an order-by-order scenario. Instead, applying the alternative block-diagonalization method via the expansion in powers of the Planck constant $\hbar$ \cite{Silenko2003,Bliokh2005,Goss2007,Goss2007b} might provide a better route to investigate the quantum corrections arising from zitterbewegung (which is responsible for the Darwin term) and operator orderings. Furthermore, as we have remarked that it might not be legitimate to block-diagonalize the Dirac or Dirac-Pauli Hamiltonian in strong fields except for special cases, the method of expansion in $\hbar$ \cite{Silenko2008} may help to elucidate the breakdown of particle-antiparticle separation in strong fields (also see \cite{Silenko2015}).

\begin{acknowledgments}
 D.W.C.\ would like to thank Sang Pyo Kim for valuable discussions. D.W.C.\ was supported in part by the Ministry of Science and Technology (Taiwan) under the Grants No.\ 101-2112-M-002-027-MY3 and No.\ 101-2112-M-003-002-MY3, and T.W.C. under the Grant No.\ 101-2112-M-110-013-MY3.
\end{acknowledgments}

\appendix*

\section{Useful formulae and lemmas}
The Pauli matrices satisfy the identity
\begin{equation}\label{identity 1}
(\boldsigma\cdot\mathbf{a})(\boldsigma\cdot\mathbf{b}) =\mathbf{a}\cdot\mathbf{b}+i(\mathbf{a}\times\mathbf{b})\cdot\boldsigma
\end{equation}
for arbitrary vectors $\mathbf{a}$ and $\mathbf{b}$.
Meanwhile, we have
\begin{equation}\label{identity 2}
(\boldsymbol{\nabla}\times\mathbf{a}+\mathbf{a}\times\boldsymbol{\nabla})\psi
=(\boldsymbol{\nabla}\times\mathbf{a})\psi.
\end{equation}
By \eqnref{identity 1} and \eqnref{identity 2}, we have
\begin{equation}\label{identity 3}
(\sdotp)^2=\boldpi^2-\frac{q\hbar}{c}\,\sdotB.
\end{equation}

Consider the commutator between $\phi$ and $\sdotp$. We have
\begin{equation}
[\phi,\,\sdotp]=i\hbar(\boldsigma\cdot\boldsymbol{\nabla})\phi=i\hbar(\sdotE),
\end{equation}
and consequently
\begin{eqnarray}
&&\left[\phi,\,(\sdotp)^2\right]=\sdotp[\phi,\sdotp]+[\phi,\sdotp]\sdotp\nonumber\\
&=&i\hbar\left[(\sdotp)\cdot(\sdotE)+(\sdotE)\cdot(\sdotp)\right]\nonumber\\
&=&i\hbar\left[\boldpi\cdot\mathbf{E}+\Edotp
+i\left(\left(\frac{\hbar}{i}\boldsymbol{\nabla}-\frac{q}{c}\mathbf{A}\right)\times\mathbf{E} +\mathbf{E}\times\left(\frac{\hbar}{i}\boldsymbol{\nabla}-\frac{q}{c}\mathbf{A}\right)\right)
\cdot\boldsigma
\right]\nonumber\\
&=&i\hbar(\boldpi\cdot\mathbf{E}+\Edotp)=2i\hbar(\Edotp),
\end{eqnarray}
where we have applied the identities \eqnref{identity 1} and \eqnref{identity 2} and assumed $\mathbf{E}$ is homogeneous.

As we consider only the terms linear in $\mathbf{E}$ and $\mathbf{B}$, we neglect the second term in \eqnref{identity 3} whenever it is multiplied by the terms containing $\mathbf{E}$ or $\mathbf{B}$. Consequently, by induction, we have
\begin{subequations}\label{identity 4}
\begin{eqnarray}
\label{identity 4a}
\left[\phi,\,(\sdotp)^{2n}\right]&=&(2n)i\hbar\,\boldpi^{2(n-1)}(\Edotp),\\
\label{identity 4b}
\left[\phi,\,(\sdotp)^{2n+1}\right] &=& i\hbar\,\boldpi^{2n}(\sdotE) +(2n)i\hbar\,\boldpi^{2n-2}(\sdotp)(\Edotp).
\end{eqnarray}
\end{subequations}

%

\end{document}